\documentclass{article}
\usepackage[margin=1.0in]{geometry}
\usepackage[sort&compress]{natbib}
\usepackage{algorithm}
\usepackage{algorithmicx}
\usepackage[noend]{algpseudocode}

\usepackage{amsmath, amsfonts, bm, amssymb, tabularx}
\usepackage{latexsym, mathtools}
\usepackage{amsthm}
\usepackage{enumitem}

\usepackage{thmtools} 
\usepackage{thm-restate}

\usepackage{ifthen}
\usepackage{hyperref}\hypersetup{colorlinks=true, unicode=true, linkcolor=[rgb]{0.10,0.05,0.67}, citecolor=[rgb]{0.10,0.05,0.67}, filecolor=[rgb]{0.10,0.05,0.67}, urlcolor=[rgb]{0.10,0.05,0.67}}
\usepackage{algorithm}
\usepackage{Shorthands}
\usepackage{cleveref}

\title{\LARGE \bf
Active Learning for Control-Oriented Identification \\of Nonlinear Systems
}

\author{Bruce D. Lee, Ingvar Ziemann, George J. Pappas, Nikolai Matni
\thanks{The authors are with the Department of Electrical and Systems Engineering, University of Pennsylvania. Emails: \tt\small\{brucele, ingvarz, pappasg, nmatni\}@seas.upenn.edu.}}

\begin{document}

\maketitle
\thispagestyle{empty}

\begin{abstract}
Model-based reinforcement learning is an effective approach for controlling an unknown system. It is based on a longstanding pipeline familiar to the control community in which one performs experiments on the environment to collect a dataset, uses the resulting dataset to identify a model of the system, and finally performs control synthesis using the identified model. As interacting with the system may be costly and time consuming, targeted exploration is crucial for developing an effective control-oriented model with minimal experimentation.  Motivated by this challenge, recent work  has begun to study finite sample data requirements and sample efficient algorithms for the problem of optimal exploration in model-based reinforcement learning. However, existing theory and algorithms are limited to model classes which are linear in the parameters.  Our work instead focuses on models with nonlinear parameter dependencies, and presents the first finite sample analysis of an active learning algorithm suitable for a general class of nonlinear dynamics. In certain settings, the excess control cost of our algorithm achieves the optimal rate, up to  logarithmic factors. We validate our approach in simulation, showcasing the advantage of active, control-oriented exploration for controlling nonlinear systems.
\end{abstract}

\section{Introduction}

In recent years, model-based reinforcement learning has been successfully applied to various application domains including robotics, healthcare, and autonomous driving \citep{levine2020offline, moerland2023model}. These approaches often proceed by performing experiments on a system to collect data, and then using the data to fit models for the dynamics. 
In the specified application domains, performing experiments requires interaction with the physical world, which can be both costly and time-consuming. 
It is therefore important to design the experimentation and identification procedures to efficiently extract the most information relevant to control. 
In particular, experiments must be designed with the downstream control objective in mind. This fact is well-established in classical controls and identification literature \citep{ljung1998system,gevers1993towards,hjalmarsson1996model, pukelsheim2006optimal}.
While these works provide some guidance for experiment design, they mostly focus on linear systems, and supply only asymptotic guarantees. 

Driven by the empirical success of machine and deep learning in solving classes of complex control problems, the learning and control communities have recently begun revisiting the classical pipeline of identification to control, proposing new algorithms, and analyzing them from a non-asymptotic viewpoint. Early efforts focused on end-to-end control guarantees for unknown linear system under naive exploration (injecting white noise inputs) \citep{dean2020sample,mania2019certainty}. These methods have also been refined by using active learning to collect better data for control synthesis
\citep{wagenmaker2021task}. This approach has been extended to nonlinear systems with a linear dependence on the unknown parameters \citep{wagenmaker2023optimal}.
Other works studying model-based control of nonlinear systems also assume linear dependence on the unknown parameters, or consider related simplifying assumptions in settings including tabular or low-rank Markov Decision Processes \citep{uehara2021pessimistic, song2021pc}. 
Model-based reinforcement learning for a general class of nonlinear systems has also been considered \citep{sukhija2023optimistic}. However, their guarantees focus on the worst case uncertainty of any control policy rather than end-to-end control costs for a particular objective. 

There is a significant gap in that there are no algorithms with strong guarantees (achieving the optimal rates) for model-based reinforcement learning of general nonlinear dynamical systems. We leverage recently developed machinery for non-asymptotic analysis of nonlinear system identification to tackle this problem \citep{ziemann2022learning}.

\subsection{Contribution}

We introduce and analyze the Active Learning for Control-Oriented Identification ($\texttt{ALCOI}$) algorithm. This algorithm extends an approach for model-based reinforcement learning proposed by \citet{wagenmaker2023optimal} for dynamical systems with a linear dependence on the unknown parameter to general nonlinear dynamics that satisfy some smoothness assumptions. The algorithm is inspired by a reduction of the excess control cost to the system identification error, which may then be controlled using novel finite sample system identification error bounds for smooth nonlinear systems.

Leveraging the aforementioned reduction of the excess control cost and system identification error bounds, we derive finite sample bounds for the excess cost of our algorithm.
\begin{theorem}[Main Result, Informal]
    Let the \texttt{ALCOI} algorithm interact with an unknown nonlinear dynamical system for some number of exploration rounds before proposing a control policy designed to optimize some objective. The excess cost of the proposed policy on the objective satisfies 
    \begin{align*}
        \mathsf{excess\,cost} \!\leq \!\frac{\mathsf{hardness \,of\, control} \!\times\! \mathsf{hardness\, of\, identification}}{\mathsf{\#\, exploration\,rounds}}.
    \end{align*}
\end{theorem}
The ``hardness of control'' captures how the error in estimation of the dynamics translates to error in control, while the ``hardness of identification'' captures how challenging it is to identify the parameters under the best possible exploration policy. 
Moreover, our analysis reveals how the respective hardness quantities interact. 
\citet{wagenmaker2023optimal} provide upper and lower bounds for this problem in a setting where the dynamics model is linear in the unknown parameters. Our upper bound is tight up to logarithmic factors in this setting, and we conjecture that it is also tight up to logarithmic factors in the setting where the model is nonlinear in the  parameters.

The non-asymptotic system identification result may be of independent interest. It derives from invoking recently developed machinery for the analysis of nonlinear system identification along with the delta method, a classical approach from statistics. These bounds provide rates that match the asymptotic limit up to logarithmic factors. 

\subsection{Related Work}

\paragraph{Additional Work Analyzing Identification \& Control}

Finite sample guarantees for active exploration of pure system identification have been studied in linear \citep{wagenmaker2020active}, and nonlinear (with linear dependence on the unknown parameters) settings \citep{mania2020active}. Lower bounds complementing the upper bounds for the end-to-end control are also present \citep{wagenmaker2021task, wagenmaker2023optimal}, and have been specialized to the linear-quadratic regulator setting  to characterize systems which are hard to learn to control \citep{lee2023fundamental}. Recent literature considers gradient-based approaches for experiment design in linear-quadratic control \citep{anderson2024control}. For more details on finite sample analysis of learning to control, see the survey by \citet{tsiamis2023statistical}. The aforementioned results do not focus on general nonlinear systems. Such analysis exists for identification; however, in the absence of end-to-end control error bounds \citep{sattar2022non,ziemann2022learning}. In contrast, we achieve end-to-end control error bounds for active learning applied for learning to control general nonlinear systems.

\paragraph{Dual Control}
A related paradigm to the ``identify then control'' scheme studied in this work is that of \emph{dual control}, in which the learner must interact with an unknown system while simultaneously optimizing a control objective  \citep{feldbaum1960dual}.  \citet{aastrom1973self} study a version of this problem known as the self-tuning regulator, providing asymptotic guarantees of convergence. Non-asymptotic guarantees for the self tuning regulator have been studied more recently from the online learning perspective of regret \citep{abbasi2011improved}. Subsequent work provides matching upper and lower bounds for the regret of the self-tuning regulator problem \citep{simchowitz2020naive}. Lower bounds refining the dependence on system-theoretic constants have also been established \cite{ziemann2022regret}. 
The regret of learning to control nonlinear dynamical systems (with linear dependence on the unknown parameter) has also been studied   \citep{kakade2020information,boffi2021regret}. 
As in the ``identify then control'' setting, prior work in dual control has not provided finite sample analysis of the end-to-end control error for systems with nonlinear dependence on the unknown parameters.
\,\\
\noindent \textbf{Notation:}
Expectation (respectively probability) with respect to all the randomness of the underlying probability space is denoted by $\E$ (respectively $\P$).
The Euclidean norm of a vector $x$ is denoted $\norm{x}$. For a matrix $A$, the spectral norm is denoted $\norm{A}$, and the Frobenius norm is denoted $\norm{A}_F$. A symmetric, positive semi-definite matrix $A = A^\top$ is denoted $A \succeq 0$.  $A \succeq B$ denotes that $A-B$ is positive semi-definite. Similarly, a symmetric, positive definite matrix $A$ is denoted $A \succ 0$. The minimum eigenvalue of a symmetric, positive semi-definite matrix $A$ is denoted $\lambda_{\min}(A)$. For a positive definite matrix $A$, we define the $A$-norm as $\norm{x}_A^2 = x^\top A x$. 
  The gradient of a scalar valued function $f: \R^n \to \R$ is denoted $\nabla f$, and the Hessian is denoted $\nabla^2 f$. The Jacobian of a vector-valued function $g: \R^n \to \R^m$ is denoted $D g$, and follows the convention for any $x\in\R^n$, the rows of $D g(x)$ are the transposed gradients of $g_i(x)$. The $p^{th}$ order derivative of $g$  is  denoted by $D^{p} g$. Note that for $p \geq 2$, $D^{p} g(x)$ is a tensor for any $x\in\R^{n}$. The operator norm of such a tensor is denoted by $\norm{D^p g(x)}_{\op}$. 
 For a function $f: \mathsf{X} \to \R^{\dy}$, we define $\norm{f}_{\infty} \triangleq \sup_{x \in \mathsf{X}} \norm{f(x)}$. A Euclidean norm ball of radius $r$ centered at $x$ is denoted $\calB(x,r)$. 
\section{Problem Formulation}
\label{s: problem formulation}

We consider a nonlinear dynamical system evolving according to 
\begin{align}
    \label{eq: dyn}
    X_{t+1} = f(X_t, U_t; \phi^\star) + W_t, \quad t=1,\dots T,
\end{align}
with state $X_t$ assuming values in $\R^{\dx}$,  input $U_t$ assuming values in $\R^{\du}$, and $\dx$-dimensional noise $W_t \overset{\iid}{\sim} \calN(0, \sigma_w^2 I)$ for some $\sigma_w > 0$. For simplicity, we assume $X_1 = 0$. Here, $f$ is the dynamics function, which depends on an unknown parameter $\phi^\star \in \R^{d_{\phi}}$. We assume that there exists some positive $B$ such that 
$\norm{\phi^\star} \leq B - 1$ and $\norm{f(\cdot, \cdot, \phi)}_{\infty} \leq B$ for all $\phi\in\R^{d_{\phi}}$ satisfying $\norm{\phi}\leq B$. 

We study a learner whose objective is to determine a policy $\pi = \curly{\pi_t}_{t=1}^T$ from a policy class $\Pi^\star$ to minimize the cost $\calJ(\pi, \phi^\star)$, where
\begin{align}
    \label{eq: objective}
    \calJ(\pi, \phi) = \E_{\pi}^{\phi} \brac{\sum_{t=1}^T c_t(X_t, U_t) + c_{T+1}(X_{T+1})}.
\end{align}
The functions $c_t$ are known stage costs. 
The superscript on the expectation denotes that the dynamics \eqref{eq: dyn} are rolled out under parameter $\phi$, while the subscript denotes that the system is played in closed-loop under the feedback control policy $U_t = \pi_t(X_1, U_1, \dots, X_{t-1}, U_{t-1}, X_t)$ for $t=1,\dots,T$. The learner follows a two step interaction protocol with an exploration phase, and an evaluation phase. In the exploration phase, the learner interacts with the system for a total of $N$ episodes, each consisting of $T$ timesteps, by playing exploration policies $\pi \in\Pi_{\exp}$. The policy class $\Pi_{\exp}$ is an exploration policy class, described in more detail below. The learner does not incur any cost during the exploration episodes, and seeks only to gain information about the system. After the $N$ interaction episodes, it uses the collected data to propose a policy $\hat\pi \in \Pi^\star$. The learner is then evaluated on the expected cost of the proposed policy on a new evaluation episode. In particular, it incurs cost $\calJ(\hat \pi, \phi^\star)$. 

The policy classes $\Pi^\star$ and $\Pi_{\mathsf{exp}}$ are known; $\Pi^\star$ consists of deterministic policies, but $\Pi_{\mathsf{exp}}$ may be random. We do not assume $\Pi^\star = \Pi_{\mathsf{exp}}$.  We assume that the policy class $\Pi^\star$ has the parametric form:
\begin{align*}
    \Pi^\star = \curly{\pi^\theta \vert \theta \in \mathbb{R}^{d_{\theta}}}.
\end{align*}
No such parametric assumption is made on the exploration class $\Pi_{\exp}$. It instead  consists of whatever experimental procedures are available. For instance, it may be the class of policies with power or energy bounded inputs. Policies belonging to $\Pi_{\exp}$ must be history dependent (causal). We assume that the learner is allowed to randomly select choices of policies in $\Pi_{\exp}$.\footnote{i.e. for any $\pi^1, \pi^2\in\Pi_{\exp}$ and any $b \in [0,1]$, the policy $\pi_{\mix}$ which at the start of a new episode plays $\pi^1$ for the duration of the episode with probability $b$ and $\pi^2$ for the duration of the episode with probability $1-b$ also belongs to $\Pi_{\exp}$. }
Given these policy classes, the learning procedure should seek to identify the best exploitation policy belonging to $\Pi^\star$ by playing the most informative exploration policy in the class $\Pi_{\exp}$.

\subsection{Certainty Equivalent Control}

\begin{figure}
    \centering
    \includegraphics[width=\linewidth]{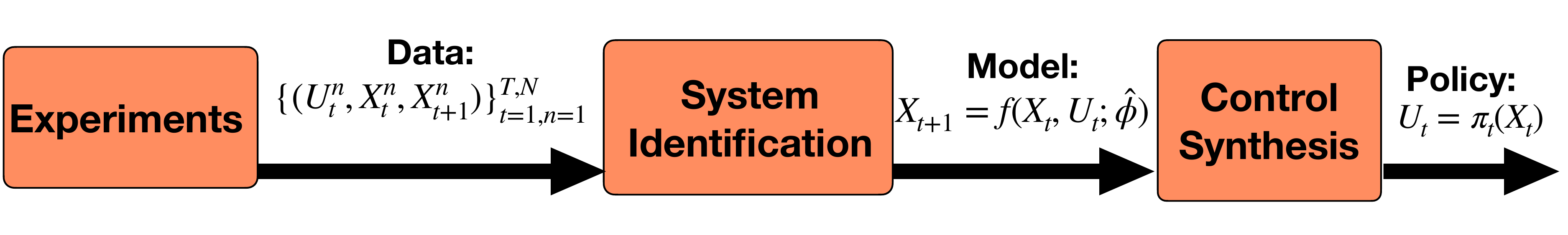}
    \vspace{-18pt}
    \caption{Identfication to control pipeline. }
    \vspace{-18pt}
    \label{fig:CE pipeline}
\end{figure}

We focus on learners which follow a model-based approach to synthesize a control policy from interaction with the system, outlined in \Cref{fig:CE pipeline}. In this section, we discuss the learner's procedure for the last two steps: system identification and control synthesis.
In \Cref{s: main}, we return to the question of which experiments the learner should perform. 

Given the data collected during the experimentation phase, the learner finds an estimate for the dynamics by solving a nonlinear least squares problem. In particular, suppose that during the $N$ experimentation episodes of length $T$, the learner collects data $\curly{(U_t^n, X_t^n, X_{t+1}^n)}_{n,t=1}^{N, T+1}$. The subscript denotes the time index within each episode, while the superscript denotes the episode index. Using this dataset, the learner may identify the dynamics of the system by solving
\begin{align}
    \label{eq: least squares}
    \hat \phi \in \argmin_{\phi \in \R^{d_{\phi}}, \norm{\phi}\leq B} \sum_{n=1}^N \sum_{t=1}^T \norm{X_{t+1}^n - f(X_t^n, U_t^n; \phi)}^2. 
\end{align}
Solving this problem provides a parameter estimate which is an effective predictor under the distribution of states and inputs seen during the experimentation. This notion can be captured via the prediction error. 

\begin{definition}
    We define $\ER_{\pi}^{\phi^\star}(\phi)$ as the prediction error for a parameter $ \phi$ under policy $\pi$: 
\begin{align*}
    \ER_{\pi}^{\phi^\star}( \phi) = \mathbb{E}_{\pi}^{\phi^\star} \left[ \frac{1}{T}\sum_{t=1}^T \norm{ f(X_t, U_t;  \phi) - f(X_t, U_t; \phi^\star) }^2 \right].
\end{align*}
\end{definition}

Once the learner estimates $\hat\phi$, the controller parameters are determined from the dynamics parameters by solving the policy optimization problem as
\begin{align}
    \theta^\star(\hat \phi) \in \argmin_{\theta \in \Re^{\dtheta}} J(\pi^\theta, \hat\phi). 
\end{align}
The certainty equivalent policy may then be expressed as a function of the estimated dynamics parameters $\hat\phi$ as
\begin{align}
    \label{eq: certainty equivalent policy}
    \pi^\star(\hat\phi) \triangleq\pi^{\theta^\star(\hat\phi)}.
\end{align}
We note that both the nonlinear least squares problem \eqref{eq: least squares} and the certainty equivalent control synthesis procedure of \eqref{eq: certainty equivalent policy} may be computationally challenging. The focus of this work is to understand the statistical complexity of the problem rather than the computational complexity. In the episodic setting we consider, both of these problems are solved offline. Therefore, given sufficient time and compute, it is often possible to determine good approximations to the optimal solutions using non-convex optimization solvers and approaches for policy optimization from the model-based reinforcement learning literature \citep{levine2020offline}.

\subsection{Assumptions}

By \eqref{eq: certainty equivalent policy}, the optimal policy for the objective \eqref{eq: objective} under the true parameter $\phi^\star$ defining the dynamics \eqref{eq: dyn} is thus given by $\pi^\star(\phi^\star)$, and the corresponding objective value is $\calJ(\pi^\star(\phi^\star), \phi^\star)$. Meanwhile, the objective value attained under an estimate $\hat \phi$ is $\calJ(\pi^\star(\hat \phi), \phi^\star)$. We abuse notation and define the shorthand 
\begin{align}
    \label{eq: shorthand cost}
    \calJ_{\tilde\phi}(\phi) \triangleq \calJ(\pi^\star(\phi), \tilde\phi)
\end{align}
to describe the control cost of applying a certainty equivalence policy synthesized using  parameter $\phi$ on a system with dynamics described by $\tilde\phi$.  
It has been shown by \citet{wagenmaker2021task, wagenmaker2023optimal} that for models which are linear in the parameters, the gap $\calJ_{\phi^\star}(\phi) - \calJ_{\phi^\star}(\phi^\star)$ is characterized by the squared parameter error weighted by the \emph{model-task Hessian}, defined below.
\begin{definition}
    The \textit{model-task Hessian} for objective \eqref{eq: objective} and dynamics \eqref{eq: dyn} is given by 
    \begin{align*}
    \calH(\tilde \phi) = \nabla_{\phi}^2 \calJ_{\tilde\phi}(\phi)\vert_{\phi=\tilde\phi}, 
    \end{align*}
    where $\calJ_{\tilde\phi}$ is defined in \eqref{eq: shorthand cost}. 
\end{definition}
To express the excess cost achieved by a certainty equivalent controller synthesized using the estimated model parameters $\hat \phi$, we operate under the following smoothness assumption on the dynamics. 

\begin{assumption}(Smooth Dynamics)
    \label{asmp: smooth dynamics}
    The dynamics are four times differentiable with respect to $u$ and $\phi$. Furthermore, for all $(x,u) \in \R^{\dx} \times \R^{\du}$, all $\phi \in \R^{d_{\phi}}$, and $i, j \in \curly{0,1,2,3}$ such that $1\leq i+j \leq 4$, the derivatives of $f$ satisfy  
    \begin{align*}
        \norm{ D_{\phi}^{(i)} D_{u}^{(j)} f(x,u;\phi)}_\op \leq L_f.
    \end{align*}
\end{assumption}
The above assumption is satisfied for, e.g., control-affine dynamics which depend smoothly on $\phi$: $f(x_t, u_t; \phi) = g_1(x_t;\phi) + g_2(u_t; \phi) u$, with $g_1$ and $g_2$ each three time differentiable with respect to $\phi$. In this example, differentiability with respect to $u$ is immediate from the affine dependence. 

We also require that the policy class $\Pi^\star$ is smooth. 
\begin{assumption}(Smooth Policy Class)
    \label{asmp: smooth policy class}
    For $t=1,\dots, T$, $x \in \calX$, and any policy $\pi \in\Pi^\star$, the function $\pi_t^\theta(x)$ is four-times differentiable in $\theta$. Furthermore, $\norm{D_{\theta}^{(i)} \pi_t^\theta(x)}_{\op} \leq L_{\theta}$ for $i=1,\dots, 4$, $\theta \in \Re^{\dtheta}$, $x \in \calX$, and $t=1,\dots, T$. 
\end{assumption}
Note that such smoothness conditions are not imposed for the exploration policy class $\Pi_{\exp}$. The exploration policy class could, for instance, consist of model predictive controllers with constraints on the injected input energy, which do not satisfy such smoothness assumptions.

We additionally require that the costs are bounded for policies in the class $\Pi_{\star}$ and all dynamics parameters in a neighborhood of the true parameter. 
\begin{assumption}(Regular costs)
    \label{asmp: bounded costs}
    The stage costs $c_t$ are three times differentiable and $\norm{D_{u}^{(i)} c_t(x,u)}_{\mathsf{op}} \leq L_{\cost}$ for $i=1,\dots, 3$, $\theta\in\R^{\dtheta}$, $(x,u)\in \R^{\dx+\du}$.
    There exists some $r_{\cost}(\phi^\star) > 0$ such that for all $\phi \in \calB(\phi^\star, r_{\cost}(\phi^\star))$, and all $\pi \in \Pi^\star$, we have $\Ex_{\pi}^{\phi}\brac{\paren{\sum_{t=1}^T c_t(X_t, U_t) + C_{T+1}(X_T)}^2} \leq L_{\cost}$.
\end{assumption}


We additionally assume that the certainty equivalent policy is a smooth function of the dynamics parameter in a neighborhood around the optimal parameter. 
\begin{assumption}
    \label{asmp: smooth CE}
    There exists some $r_{\theta}(\phi^\star) > 0$ such that for all $\phi \in \calB(\phi^\star, r_{\theta}(\phi^\star))$,
    \begin{itemize}[leftmargin=*]
        \item $\nabla_{\theta} \calJ(\pi^
        \theta, \phi)\vert_{\theta = \theta^\star(\phi)}=0$ 
        \item $\theta^\star(\phi)$ is three times differentiable in $\phi$ and  $\norm{D_{\phi}^i \theta^\star(\phi)}_{\op} \!\leq\! L_{\pi^\star}$ for some $L_{\pi^\star} > 0$ and $i \in \curly{1,2,3}$. 
    \end{itemize}
\end{assumption}
It is shown in Proposition 6 of \citet{wagenmaker2023optimal} that the above condition holds if $\nabla_\theta^2 J(\pi^\theta, \phi^\star) \succ 0$.\footnote{\citet{wagenmaker2023optimal} show this result for the linear in the parameters setting; however, it extends easily to the smooth nonlinear setting. See \Cref{s: smooth systems}.} The above assumption also holds in the setting of linear-quadratic regulation, as may be verified using the LQR derivative expressions in \citet{simchowitz2020naive}.    

In order to bound the parameter recovery error in terms of the prediction error, additional identifiability conditions are needed. The following definition of a Lojasiewicz exploration policy is determined from a Lojasiewicz condition that arises in the optimization literature that measures the sharpness of an objective near its optimizer \citep{roulet2017sharpness}. In our setting, it quantifies the degree of identifiability from using a particular exploration policy. It does so by bounding the growth of identification error as a polynomial of prediction error. 
\begin{definition} [Lojasiewicz condition, \citet{roulet2017sharpness}]\label{def: loja}
For positive numbers $C_{\loja}$ and $\alpha$, we say that a policy $\pi\in\Pi_{\exp}$ is $(C_{\loja}, \alpha)$-Lojasiewicz if 
\begin{align*}
    \norm{\hat \phi - \phi^\star} \leq C_{\loja} \ER_{\pi}^{\phi^\star}(\hat \phi)^\alpha.
\end{align*}
\end{definition}


To ensure parameter recovery is possible for the learner, we make the following assumption regarding identifiability. 
\begin{assumption}
    \label{asmp: loja}
    Fix some positive constant $C_{\loja}$ and $\alpha \in (\frac{1}{4}, \frac{1}{2}]$. The learner has access to a policy $\pi^0 \in \Pi_{\exp}$ which is $(C_{\loja}, \alpha)$-Lojasiewicz.
\end{assumption}

While the Lojasiewicz assumption ensures that the data collected via the exploration policy $\pi^0$ is sufficient to identify the parameters, the rate of recovery may be slow. To bypass this limitation, we assume that some policy in the exploration class  satisfies a persistence of excitation condition. This condition can be expressed 
by first defining the Fisher information matrix for a parameter $\phi$ and a policy $\pi$ as
 \begin{align}
        \label{eq: fisher info}\FI^\pi(\phi)\!\triangleq\!\frac{\E_{\pi}^{\phi} \brac{\sum_{t=1}^T \!D f(X_t, U_t, \phi)^\top  \!D f(X_t, U_t, \phi) }}{\sigma_w^2},
\end{align}
where $D f(X_t, U_t, \phi)$ is the Jacobian of $f$ with respect to $\phi$.
The Fisher information measures the signal-to-noise ratio of the data collected from an episode of interaction with the system under exploration policy $\pi$. 
With this definition, persistance of excitation is equivalent to the positive definiteness of the matrix $\FI^\pi(\phi^\star)$.

\begin{assumption}
    \label{asmp: good policy}
    There exists a policy $\pi \in \Pi_{\exp}$ for which
    \begin{align*}
        \FI^\pi(\phi^\star) \succeq \mu I \succ 0.
    \end{align*}
\end{assumption}



\section{Proposed Algorithm and Main Result}
\label{s: main}

The above smoothness assumptions allow us to characterize the excess control cost of a policy synthesized via  certainty equivalence  applied to  a parameter estimate $\hat\phi$, \eqref{eq: certainty equivalent policy}. In particular, we extend a result from \citet{wagenmaker2023optimal} from the linear in parameters setting to the smooth nonlinear setting. 
\begin{lemma}[Thm. 1 of  \citet{wagenmaker2023optimal}]
    \label{lem: cost decomposition}
    Suppose Assumptions~\ref{asmp: smooth dynamics}-\ref{asmp: smooth CE} hold. Let $r_{\theta}(\phi^\star)$ be as defined in Assumption~\ref{asmp: smooth CE}.
    Then for $\hat\phi \in \calB(\phi^\star, \min\curly{r_{\cost}(\phi^\star), r_{\theta}(\phi^\star)})$, 
    \begin{align}
            \label{eq: excess cost as parameter error}
            \calJ_{\phi^\star}(\hat\phi)\! -\! \calJ_{\phi^\star}(\phi^\star)\! \leq \!\norm{\hat \phi \!-\! \phi^\star}^2_{\calH(\phi^\star)} +  C_{\cost} \norm{\hat \phi \!-\! \phi^\star}^3,
        \end{align}
        where $\calJ_{\phi^\star}(\phi)$ is as defined in \eqref{eq: shorthand cost} and
        \begin{align*}
            C_{\cost} = \mathsf{poly}(L_{\pi^\star},  L_f, L_{\theta}, L_{\cost}, \sigma_w^{-1}, T, \dx).
        \end{align*}
\end{lemma} 

\Cref{lem: cost decomposition} informs us that the leading term of the excess cost is given by the parameter estimation error weighted by the model-task Hessian, $\norm{\hat\phi-\phi^\star}_{\calH(\phi^\star)}^2$. 

Asymptotically, the distribution of the parameter estimation error is normally distributed with mean zero, and covariance given by the inverse Fisher information matrix under the data collection policy $\pi$ evaluated at the true parameter value (cf. Theorem 1 of  \citet{lennart1980asymptotic}):
    \begin{align*}
        \lim_{N\to\infty} \sqrt{N} \FI^\pi(\phi^\star)^{1/2} (\hat \phi - \phi_\star) \sim \calN(0, I).
    \end{align*}

We provide a novel non-asymptotic result which characterizes the $H$-norm of the parameter error for a positive definite matrix $H$ in terms of the Fisher information matrix. 
\begin{theorem}
    \label{thm: identification error bound}
    Suppose Assumption~\ref{asmp: smooth dynamics} holds.  Consider the least squares estimate $\hat\phi$ determined from \eqref{eq: least squares} using data collected from $N$ episodes via an exploration policy $\pi$ which is $(C_{\loja}, \alpha)$-Lojasiewicz for some $\alpha \in (\frac{1}{4}, \frac{1}{2}]$, and satisfies $\lambda_{\min}\paren{\FI^\pi(\phi^\star)} > 0$, with $\FI^\pi(\phi^\star)$ as defined in \eqref{eq: fisher info}. Let $H$ be a positive definite matrix, $\beta$ a positive number satisyfing
        $\beta \leq \sigma_w^2 \frac{\lambda_{\min}(\FI^\pi(\phi^\star))}{4},$ and $\delta \in (0, \frac{1}{4}]$. Then there exists a polynomial $\mathsf{poly}_{\alpha}$ depending on $\alpha$ such that the following condition holds. With probability at least $1-\delta$, 
    \begin{equation}
    \begin{aligned}
        \label{eq: id error}
         &\norm{\hat \phi - \phi^\star}_H^2 \leq   2 (1+\xi) \\&\times \paren{ \frac{\trace(H \FI^\pi(\phi^\star)^{-1}) }{N} + 2 \frac{\norm{H \FI^\pi(\phi^\star)^{-1}} }{N} \log\frac{4}{\delta}},
    \end{aligned}
    \end{equation}
    where $\xi = 4\beta \paren{\frac{1}{\sigma_w^2 \lambda_{\min}(\FI^\pi(\phi^\star))} + d_{\phi}}$
    as long as
    \begin{align*}
        N &\geq \mathsf{poly}_\alpha\bigg(T,L_f,  \!d_{\phi},\! \dx, \! \sigma_w, \!\log N, \! \log\frac{1}{\delta}, \!\log \frac{B}{\sigma_w} , \! C_{\loja}, \!\frac{1}{\beta} \bigg).
    \end{align*}
\end{theorem}

By substituting the inequality \eqref{eq: id error} into the leading term of \eqref{eq: excess cost as parameter error} and bounding $\norm{
\calH(\phi^\star) \FI^\pi(\phi^\star)^{-1}} \leq \trace(\calH(\phi^\star) \FI^\pi(\phi^\star)^{-1})$, one would expect that the the excess cost of deploying the certainty equivalent policy synthesized on a least squares estimate determined from data collected using policy $\pi$ is characterized by 
\begin{align}
    \label{eq: error charcterizing quantity}\trace(\calH(\phi^\star) \FI^\pi (\phi^\star)^{-1}).
\end{align}
In light of this, we would like to choose the exploration policy $\pi$ which minimizes this upper bound:
\begin{align}
    \label{eq: ideal experiment design}
    \pi = \argmin_{\tilde \pi \in \Pi_{\exp}} \trace(\calH(\phi^\star) \FI^{\tilde \pi} (\phi^\star)^{-1}).
\end{align}
Our main result  shows that the above intuition can be made rigorous through careful analysis and design of the exploration policy. It then proceeds to show that we can find an exploration policy approximately solving \eqref{eq: ideal experiment design}, even though the parameter $\phi^\star$ defining the exploration objective is unknown prior to experimentation. It is also necessary to address the fact that the model-task Hessian may not be positive definite. Thus optimizing the above objective could result in exploration policies which are not persistently exciting.

To circumvent the issue of the unknown parameter $\phi^\star$, we consider a two step approach in which we first obtain a crude parameter estimate, and then refine it by playing a targeted exploration policy. Denote the crude estimate by $\hat\phi^-$.  This parameter can be used to search for a policy that approximately solves the optimization problem in \eqref{eq: ideal experiment design}. A straightforward approach to do so is to solve the problem under the estimated parameter: 
\begin{align}
    \label{eq: approximate exp design}
    \pi = \argmin_{\tilde \pi \in \Pi_{\exp}} \trace(H(\hat \phi^{-}) \FI^{\tilde \pi}(\hat \phi^-)^{-1}). 
\end{align}

To address the issue of a model-task Hessian which is not positive definite, we introduce regularization into the exploration design. In particular, we set the exploration policy $\pi$ as 
\begin{align}
    \label{eq: regularized approximate exp design}
    \pi = \argmin_{\tilde \pi \in \Pi_{\exp}} \trace((H(\hat \phi^{-}) + \nu I) \FI^{\tilde \pi}(\hat \phi^-)^{-1}),
\end{align}
for an appropriately chosen regularization parameter $\nu$.

The above discussion motivates \Cref{alg: ALCOI}, named Active Learning for Control-Oriented Identification ($\texttt{ALCOI}$). The algorithm takes as input an initial policy satisfying the Lojasiewicz condition (Assumption~\ref{asmp: loja}), the exploration policy class, the target policy, the number of exploration rounds, a parameter $\gamma \in (0, 1)$ which controls the ratio of the exploration budget that the initial loja policy is played, the level of regularization $\nu$, and the precision of the optimization for the exploration policy $\varepsilon$. Given these  components, the algorithm proceeds in three stages. The first stage begins in \Cref{line: intial policy} by playing the initial policy for a portion of the exploration budget controlled by $\gamma$. In \Cref{line: coarse estimate}, it uses the collected data to derive a coarse estimate $\hat \phi^-$ for the unknown parameters by solving a least squares problem. 
Next, the estimate $\hat \phi^-$ is used to construct the model-task Hessian as $\hat\phi^-$ as $\calH(\hat \phi^-)$ and define an exploration objective $\trace\paren{\paren{\calH(\hat\phi^-) + \lambda I}\FI^{\tilde \pi}(\hat \phi^-)^{-1}}$. Optimizing this objective to precision $\varepsilon$ over the class of exploration policies provides the policy $\pi_{\exp}$. 
This policy is run to collect data from the system, and obtain a fresh estimate $\hat\phi^+$ for $\phi^\star$. 
Finally, the estimate is used to synthesize the certainty equivalent policy as in \eqref{eq: certainty equivalent policy}.


\begin{algorithm}
\caption{$\texttt{ALCOI}(\pi^0, \Pi_{\exp}, \Pi^\star, N, \gamma, \nu, \varepsilon)$} 
\label{alg: ALCOI}
\begin{algorithmic}[1]
\State \textbf{Input:} Initial policy $\pi^0$, exploration policy class $\Pi_{\exp}$, target policy class $\Pi^\star$, initial policy ratio $\gamma$, regularization parameter $\nu$, optimization precision $\varepsilon$. 
\State \label{line: intial policy} Play $\pi^0$ for $\lfloor \gamma N  \rfloor$ episodes to collect $\curly{X_t^n, U_t^n, X_{t+1}^n}_{t, k=1}^{T, \lfloor \gamma N \rfloor}$.
\State \label{line: coarse estimate} Fit $\hat \phi^-$ from the collected data by solving \eqref{eq: least squares}.
\State \label{line: doed call} Determine exploration policy as \begin{align*} \pi_{\exp} \in\Bigg\{\pi\in\Pi_{\exp}\vert \trace\paren{\calH(\hat \phi^-) \FI^{ \pi}(\hat \phi^-)^{-1}} \\\leq (1+\varepsilon) \inf_{\tilde \pi} \trace\paren{\calH(\hat \phi^-) \FI^{\tilde \pi}(\hat \phi^-)^{-1}} \Bigg\}.\end{align*}.
\State \label{line: mix} Define $\pi_{\mix}$ which at the start of each episode plays $\pi^0$ with probability $\gamma$, and $\pi_{\exp}$ with probability $1-\gamma$.
\State \label{line: play} Play $\pi_{\mix}$ for $\lfloor (1-\gamma) N \rfloor$ episodes, collecting data $\curly{X_t^n, U_t^n, X_{t+1}^n}_{t, k=1}^{T, \lfloor (1-\gamma) N \rfloor}$.
\State \label{line: fine model} Fit $\hat \phi^+$ by solving \eqref{eq: least squares} with the data $\curly{X_t^n, U_t^n, X_{t+1}^n}_{t, k=1}^{T, \lfloor (1-\gamma)N \rfloor}$.
\State \textbf{Return: } certainty equivalent policy $\hat \pi = \pi^\star(\hat \phi^+)$.
\end{algorithmic}
\end{algorithm}


Our main result is a finite sample bound characterizing the excess cost of the policy return by \Cref{alg: ALCOI}. 
\begin{restatable}[Main Result]{theorem}{mainresult}\label{thm: main}
Suppose $f$, $\pi^0$, $\Pi_{\exp}$, $\Pi^\star$ satisfy Assumptions~\ref{asmp: smooth dynamics}-\ref{asmp: good policy}. Let $\nu$ be a non-negative regularization parameter such that $\lambda_{\min}(\calH(\phi^\star)) + \nu > 0$. Let the optimization tolerance $\varepsilon \in (0,1/2)$ and the initial policy ratio $\gamma\in(0,1/2)$. Consider running \Cref{alg: ALCOI} to generate a control policy $\hat\pi$ as $\hat \pi = \texttt{ALCOI}(\pi^0, \Pi_{\exp}, \Pi^\star, N, \gamma, \nu, \varepsilon)$. 

Let $\delta \in \big(0, \frac{1}{4}]$ be the failure probability, and $\beta \in \bigg( 0,  \frac{\mu \paren{\lambda_{\min}(H(\phi^\star)) + \nu}}{512 d_{\phi} (\norm{H(\phi^\star)} + \nu)}\bigg)$\footnote{Recall that $\mu$ is the persistance of excitation parameter defined in Assumption~\ref{asmp: good policy}.} be a free parameter in the bound. There exists a polynomial function $\texttt{poly}_\alpha$ depending on the Lojasiewicz parameter $\alpha$ such that the following holds true. 
With probability at least $1-\delta$, it holds that
\begin{align*}
    &\calJ(\hat \pi, \phi^\star) \!-\!  \calJ_{\phi^\star}(
\phi^\star) \!\leq\! (1\!+\!4\gamma)(1\!+\!\varepsilon)(1\!+\!\xi) \paren{2 + 4\log\frac{4}{\delta}}  \\&\times \frac{\inf_{\tilde \pi \in \Pi_{\exp}} \trace\paren{\paren{\calH(\phi^\star) + \nu I} \FI^{\tilde \pi}(\phi^\star)^{-1} }}{N},
\end{align*}
where $\xi =\beta\paren{3  + 16\paren{\frac{128 d_{\phi} \paren{\norm{\calH(\phi^\star)} +\nu}}{\mu \paren{\lambda_{\min}(\calH(\phi^\star)) + \nu}} + d_{\phi}}}$ as long as 
\begin{align*}
    N \!&\geq \!\mathsf{poly}_\alpha\bigg(\!T,L_f, L_{\cost},  L_{\pi^\star}, d_{\phi}, \dx, \frac{1}{\mu}, r_{\cost}(\phi^\star), r_{\theta}(\phi^\star)\\&\frac{1}{\lambda_{\min}(\calH(\phi^\star)) + \nu}, \norm{\calH(\phi^\star)}, \nu, \sigma_w, \sigma_w^{-1}, \log N, \\& \log\frac{1}{\delta}, \log B, C_{\loja}, \frac{1}{\gamma}, \frac{1}{\beta} \bigg).
\end{align*}
\end{restatable}

The above result characterizes the excess control cost in terms of three key quantities. First, the term $\calH(\phi^\star)$ is the model-task Hessian, which describes how error in identification of the dynamics model parameters translates to the control cost. Second is the inverse Fisher information of the optimal exploration policy term, which measures a signal-to-noise ratio quantifying the hardness of parameter identification. Finally, the number of exploration episodes $N$ on the denominator captures the rate of decay from increasing the experimental budget. 

In the setting where the dynamics model has linear dependence on the parameters, \citet{wagenmaker2023optimal} present a lower bound on the excess control cost achieved by any learner following the model based interaction protocol described in \Cref{s: problem formulation}. If we choose the free parameter in the upper bound as $\beta \leq \paren{\frac{128 d_{\phi} \paren{\norm{\calH(\phi^\star)} +\nu}}{\mu \paren{\lambda_{\min}(\calH(\phi^\star)) + \nu}} + d_{\phi}}^{-1}$, then the upper bound of \Cref{thm: main} matches this lower bound up to universal constants, and the term $\log\paren{4/\delta}$. Future work will pursue general lower bounds that hold for dynamics models with a nonlinear dependence on the unknown parameter. 

The burn-in time is currently polynomial in the relevant system parameters; however, we do not pursue tight burn-in times in this work. It may be possible to improve the dependence of the burn-in on various system quantities, e.g. by leveraging stability or reachability to obtain optimal dependence of the burn-in on $T$. We additionally draw attention to the utility of the parameter $\beta$ and the algorithm hyperparameters $\nu$ and $\gamma$ for navigating the tradeoff between a good burn-in time, and optimal rates. One can take $\beta$, $\nu$, $\gamma$ and $\varepsilon$ arbitrarily close to zero, meaning that the coefficient characterizing the excess cost can become arbitrarily close to $2 + 4 \log\frac{8}{\delta}$. The cost of doing so is an increase in the burn-in time. A notable exception is the situation where $\calH(\phi^\star) \succ 0$, i.e. the setting where all parameters are necessary for control. In this case, one can take $\nu = 0$, while the burn-in time must exceed a polynomial in $\frac{1}{\lambda_{\min}(\calH(\phi^\star))}$. 

\section{Proof Sketch}
\label{s: pf sketch}

Full proof details may be found in the appendix. Here, we present a sketch. Our main result proceeds by demonstrating the following sub-steps. In these sub-steps, let $C$ be a polynomial of the problem parameters and log of the reciprocal of the failure probability, $\delta$, as in the burn-in requirement of \Cref{thm: main}. 
\begin{enumerate}[noitemsep,nolistsep,leftmargin=*]
    \item \label{item: coarse id bound} With high probability, the coarse parameter estimation error decays gracefully with the total amount of data:
    \begin{align}
        \label{eq: coarse estimation error}
        \norm{\hat\phi^--\phi^\star} \leq \frac{C}{(TN)^{\alpha}},
    \end{align}
   as long as $N$ exceeds some polynomial burn-in time.  This result is derived from recent results characterizing non-asymptotic bounds for identification \citep{ziemann2022learning}, and takes the place of the estimator consistency requirements in classical asymptotic identification literature \citep{lennart1980asymptotic}. By making the number of episodes $N$ sufficiently large, we can make this error arbitrarily small. It thus characterizes a type of ``consistency burn-in''. 
   \item \label{item: coarse approx of exp design} As long as the coarse estimation error of \eqref{eq: coarse estimation error} is sufficiently small, the ideal optimal exploration objective of \eqref{eq: ideal experiment design} is well-approximated by the objective \eqref{eq: approximate exp design}. In particular, for any exploration policy $\pi\in\Pi_{\exp}$,
   \begin{equation}
   \label{eq: ideal to approx exploration}
   \begin{aligned}
       \!&\bigg|\!\trace\paren{\!\calH(\hat \phi^-) \FI^\pi(\hat \phi^-)^{-1}\!}\! -\!\trace(\calH(\phi^\star) \FI^{\pi} (\phi^\star)^{-1}) \bigg|\! \\&\leq \!C\! \norm{\hat\phi^- \!\!- \!\phi^\star}.
   \end{aligned}
   \end{equation}
   \item \label{item: sysid bnd} For $N$ sufficiently large, we may use the consistency guarantee  \eqref{eq: coarse estimation error} to prove \Cref{thm: identification error bound}. 
   The proof of this fact follows by revisiting the delta method \citep{van2000asymptotic} through the lens of concentration inequalities. Doing so results in the near sharp\footnote{Rates matching the  asymptotic limit up to logarithmic factors.} rates we obtain. 
\end{enumerate}
Using the above results, our argument proceeds according to the following series of inequalities applied to the excess cost. With high probability,
\begin{equation}
\label{eq: excess cost pf}
\begin{aligned}
    &\calJ_{\phi^\star}(\hat\phi^+)  \!-\! \calJ_{\phi^\star}(\phi^\star) \leq \norm{\hat\phi^+ - \phi^\star}_{\calH(\phi^\star)}^2 \!+\! C_{\cost} \norm{\hat\phi^+ - \phi^\star}^3 \\
    &\leq \xi(\delta) \frac{\trace(\paren{\calH(\phi^\star) + \nu I } \FI^{\pi_{\mix}}(\phi^\star)^{-1}) }{N} + \frac{C}{N^{3/2}},
\end{aligned}
\end{equation}
where the first inequality follows by \Cref{lem: cost decomposition}, and the second inequality follows by applying \Cref{thm: identification error bound} with $H=\calH(\phi^\star) + \nu I$ for the first term, and $H=I$ for the second term. The quantity $\xi(\delta)$ is a trades off the burn-in time and the final bound. In our analysis, it can become as small as $2 + 4\log\frac{1}{\delta}$. 
Next, it follows that
\begin{align*}
    &\trace(\paren{\calH(\phi^\star)+\nu I} \FI^{\pi_{\mix}}(\phi^\star)^{-1}) \\&\overset{(i)}{\leq}  \frac{C}{N^\alpha}  +  \trace\paren{\paren{H(\hat \phi^-)+\nu I} \FI^{\pi_{\mix}}(\hat\phi^-){-1}}  \\
    &\overset{(ii)}{\leq} \!  2\frac{C}{N^\alpha} \!+\! \frac{1}{1-\gamma} \inf_{\tilde\pi\in\Pi_{\exp}}\!\trace\Bigg(\!\paren{H(\hat \phi^-) \!+\! \nu I} \FI^{\tilde \pi}(\hat\phi^-)^{-1}\!\Bigg) \\
       &\overset{(iii)}{\leq} \!3\frac{C}{N^\alpha} \!+\!\frac{1}{1\!-\!\gamma}\!\inf_{\tilde\pi\in\Pi_{\exp}}\!\trace\Bigg(\paren{H(\phi^\star) \!+\! \nu I} \FI^{\tilde\pi}(\phi^\star)^{-1}\Bigg),
\end{align*}
where inequality $(i)$ follows from \eqref{eq: ideal to approx exploration} and \eqref{eq: coarse estimation error}, inequality $(ii)$ follows from the definition of the policy $\pi_{\mix}$ and inequality $(iii)$ follows from \eqref{eq: ideal to approx exploration} and \eqref{eq: coarse estimation error}. The main result then follows by substituting the above bound into \eqref{eq: excess cost pf}, and taking $N$ to exceed a polynomial burn-in time so the higher order terms become negligible.

\section{Numerical Validation}

We deploy $\texttt{ALCOI}$ on an illustrative example to illustrate the benefits of active control-oriented exploration. For more experiments, and further details, see \Cref{s: experiment details}. Consider the two dimensional system
\begin{align*}
    X_{t+1} = X_t + U_t + W_t + \sum_{i=1}^4 \psi(X_t - \phi_\star^{(i)})
\end{align*}
with $X_t, U_t, W_t$ and $\phi_\star^{(i)}$ assuming values in $\R^2$. 
Here $\psi: \R^2 \to \R^2$ is defined by $\psi(x) = 5\frac{x}{\norm{x}}\exp(-x^2)$.  The noise is distributed according to a standard normal distribution. The parameters $\phi_\star^{(1)}$,  $\phi_\star^{(2)}$, $\phi_\star^{(3)}$, $\phi_\star^{(4)}$ are set as $\bmat{5 \\ 0}, \bmat{-5 \\ 0}, \bmat{0 \\ 5}$ and $\bmat{0 \\ -5}$, respectively.

We consider model-based reinforcement learning with a horizon $T=10$ and quadratic cost functions: for all $t \in [T]$, $$c_t(x,u) = \norm{x - \bmat{5.5 \\ 0}}^2, \,  c_{T+1}(x) = \norm{x - \bmat{5.5 \\ 0}}^2.$$ The policy class $\Pi^\star$ consists of feedback linearization controllers defined by parameters $\theta = (K, \hat \phi^{(1)}, \dots, \hat \phi^{(4)})$, with $K \in \R^{2\times2}$ and $\hat\phi^{(i)} \in \R^2$ for $i=1,\dots,4$:  $$\pi^{\theta}(X_t) = K \paren{X_t - \bmat{-5.5 \\ 0}} - \sum_{i=1}^4 \psi(X_t - \hat\phi^{(i)}).$$The exploration class $\Pi_{\exp}$ consists of policies with input energy bounded by $T$: $\sum_{t=1}^T \norm{U_t}^2 \leq T$. 

We compare $\texttt{ALCOI}$ with random exploration and approximate $A$-optimal experiment design. For random exploration, the learner injects isotropic Gaussian noise which is normalized such that $\sum_{t=1}^T \norm{U_t}^2= T$. For approximate $A$-optimal experiment design, the learner runs the $\texttt{ALCOI}$, but with the model-task Hessian estimate, $\calH(\hat\phi^-)$, replaced by $I$. 

\begin{figure}
    \centering
    \includegraphics[width=0.95\linewidth]{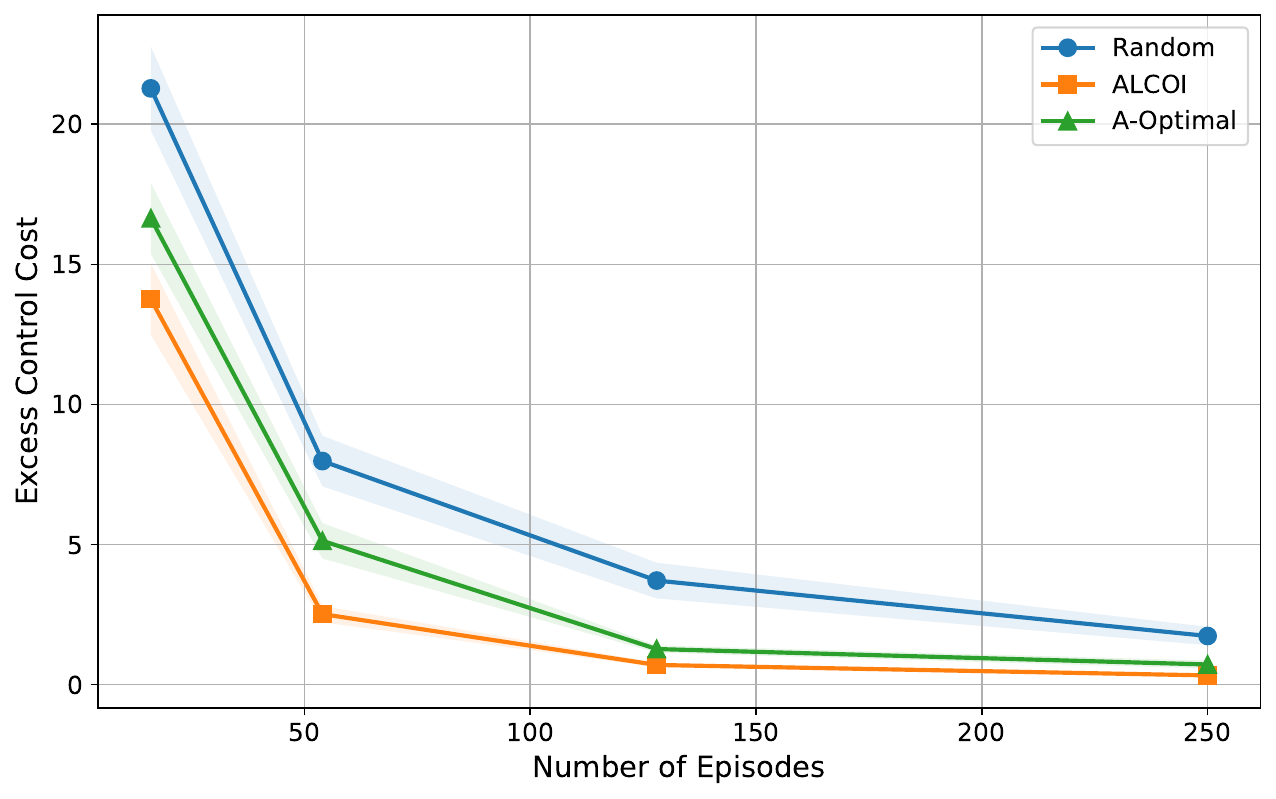}
    \caption{Comparison of the proposed control-oriented identification procedure with approximate $A$-optimal design, and random experiment design. The mean over $100$ runs is shown, with the standard error shaded.  }
    \label{fig: experiment design comparison}
    \vspace{-14pt}
\end{figure}

\Cref{fig: experiment design comparison} illustrates that \texttt{ALCOI} achieves a lower excess control cost than the alternatives at all iterations. To understand why this is the case,  note that in order to regulate the system to the position $X_t = \bmat{5.5 & 0}^\top$, the parameter $\phi_\star^{(1)}$ must be identified accurately. However, due to the Gaussian kernel, accurately estimating $\phi_\star^{(1)}$ requires that the experiment data consists of trajectories where the state is near $\phi_\star^{(1)}$. Random exploration clearly fails to collect such trajectories. Approximate $A$-optimal experiment design does collect such trajectories; however, it also collects trajectories steering the state to $\bmat{-5 & 0}^\top$, $\bmat{0 & 5}^\top$, and $\bmat{0 & -5}^\top$ in order to identify the parameters $\phi_\star^{(2)}, \phi_\star^{(3)}$ and $\phi_\star^{(3)}$. \texttt{ALCOI}, in contrast, designs experiments that are effective for identifying the parameters most relevant for control. For the chosen objective, this means that the algorithm invests the most exploration energy in collecting data in the neighborhood of $\phi_\star^{(1)}$.     This illustrative example hints at the practical benefit of the proposed approach. 

\section{Conclusions}

We have introduced and analyzed the Active Learning for Control-Oriented Identification (\texttt{ALCOI}) algorithm, marking a significant step towards understanding active exploration in model-based reinforcement learning for a general class of nonlinear dynamical systems. We provide finite sample bounds on the excess control cost achieved by the algorithm which offer insight into the interaction between the hardness of control and identification. Our bounds are known to be sharp up to logarithmic factors in the setting of nonlinear dynamical systems with linear dependence on the parameters, and we conjecture that they are sharp in general. Future work will attempt to verify that this is the case.
It would also be interesting for future work to consider learning partially observed dynamics using general prediction error methods, rather than assuming a noiseless state observation. 

\section*{ACKNOWLEDGMENT}
BL and NM are supported by NSF Award SLES-2331880, NSF CAREER award ECCS-2045834 and AFOSR Award FA9550-24-1-0102.
IZ is supported by a Swedish Research Council international postdoc grant. GP is supported in part by NSF Award SLES-2331880. 

\bibliographystyle{IEEEtranN}
\bibliography{refs}

\clearpage

\appendix
\onecolumn 

\tableofcontents

\section{System Identification Results}

In this section we consider  $\curly{X_t^k, U_t^k, X_{t+1}^k}_{t,k=1}^{T,K}$ with joint probability distribution $\mathsf{P}^K = \mathsf{P} \times \dots \times  \mathsf{P} $ induced playing an exploration policy $\pi \in \Pi_{\mathsf{exp}}$ on the system dynamics \eqref{eq: dyn}. As we consider a single dynamics parameter $\phi^\star$, and a single data collection policy $\pi$, we drop the subscripts $\pi$ and $\phi^\star$ from the expectations and probabilities. 

The identification results are organized as follows. First, in \Cref{s: consistency}, we show that under a policy satisfying a Lojasiewicz condition \Cref{def: loja}, the least squares parameter estimation error decays with the number of episodes $N$. \Cref{s: delta method} leverages the smoothness of the dynamics and the consistency result from \Cref{s: consistency} to express the parameter error in terms of a self-normalized martingale. This self-normalized martingale is then bounded in \Cref{s: proof of identification error bound} to prove \Cref{thm: identification error bound}. 

\subsection{Consistency of Least Squares Parameter Estimation}
\label{s: consistency}

Recent developments \cite{ziemann2022learning, ziemann2023tutorial} have provided finite sample bounds for the $L_2$ prediction error of nonlinear least squares estimators in the presence of time dependent data with rates that do not depend on the mixing time of the data. Our first estimation bound instantiates these results in our setting under the assumption that the data collection policy is $(C_{\loja}, \alpha)$-Lojasiewicz.  To express these bounds we define $\calG = \curly{f(\cdot, \cdot, \phi) : \phi \in \R^{d_{\phi}}, \norm{\phi}\leq B}$ and $\calG_\star = \calG - \curly{f(\cdot, \cdot, \phi^\star)}$. Using the Lipschitz bound on $f$, we find that the size  $\calN(\calG_\star, \norm{\cdot}_{\infty}, \varepsilon)$ of the the minimal cardinality $\varepsilon$-net covering $\calG_\star$ is bounded as
\begin{align}
    \label{eq: metric entropy of parametric class}
    \log \calN(\calG_\star, \norm{\cdot}_{\infty}, \varepsilon) \leq  d_{\phi} \log\paren{\frac{3 B L_f}{\varepsilon} +1}.
\end{align}

We will present an estimation error bound that holds as long as the number of episodes exceeds some value (a burn-in time) depending on the confidence with which we want it to hold. To state this requirement, define
\begin{equation}
\label{eq: ec burn-in}
\begin{aligned}
     &\tau_{\ER}(\delta) = \max\bigg\{ \paren{\frac{256 \sigma_w^2}{T} \paren{\dx + d_{\phi}\log\paren{4 B L_f TK/\sigma_w + 1} + \log\frac{1}{\delta}}}^{\frac{1}{4\alpha-1}}, \\
    &\Bigg(8(L_f C_{\loja} \alpha^{-1})^4 \paren{d_{\phi} \log \paren{4 B L_f K^{2\alpha}+1}} + \log \frac{1}{\delta} \Bigg)^{\frac{1}{1-8\alpha + 16 \alpha^2}} \bigg\}.
\end{aligned}
\end{equation}

\begin{lemma} 
\label{lem: er bound}
Suppose Assumption~\ref{asmp: smooth dynamics} holds, and $\delta\in(0,\frac{1}{2}]$. Let $\hat\phi$ be the least squares estimate from \eqref{eq: least squares} using $K$ episodes of data collected with a $(C_{\loja}, \alpha)$ policy, for $K \geq \tau_{\mathsf{ER}}(\delta)$. Then with probability at least $1-\delta$
    \begin{align*}
        \ER_{\pi, \phi^\star}(\hat \phi) \leq \frac{512 \sigma_w^2}{TK } \paren{\dx + d_{\phi}\log\paren{\frac{4L_f TK}{\sigma_w \delta}}}.
    \end{align*}
\end{lemma}

To prove this result, we first state several preliminary results. In order to do so, we introduce additional notation. Let $\calG_\star^s = \curly{bg: b\in[0,1], g\in\calG_\star}$. This set is called the star-hull of $\calG_\star$\footnote{The star-hull of $\calG$ is the smallest set containing $\calG$ that is star-shaped.}. Since all elements of $\calG_\star^s$ satisfy $\norm{f(\cdot, \cdot, \phi)}_{\infty} \leq B$, we have from Lemma 4.5 of \citet{mendelson2002improving} that the size $\calN(\calG_\star^s, \norm{\cdot}_\infty, \varepsilon)$ of a minimum cardinality $\varepsilon$-net of $\calG_\star^s$ in the $\norm{\cdot}_{\infty}$-norm satisfies  
 \begin{align}
    \label{eq: metric entropy of star hull}
     \log \calN(\calG_\star^s, \norm{\cdot}_\infty, \varepsilon) \leq \log\paren{4B/\varepsilon} + \log \calN(\calG_\star, \norm{\cdot}_\infty, \varepsilon).
 \end{align}
Define $B(r) = \curly{g \in \calG_\star^s: \sqrt{\E\brac{\frac{1}{T} \sum_{t=1}^T \norm{g(X_t^1, U_t^1)}^2}} \leq r}$, and denote the boundary of this set as $\partial B(r)$. 

We begin by stating the definition of hypercontractivity from \cite{ziemann2022learning}.
\begin{definition}[Trajectory $(C_{\hc}, a)$-hypercontractivity]
    Fix constants $C_{\hc} > 0$ and $a \in [1,2]$. The tuple $(\calG, \mathsf{P})$ consisting of the function class $\calG$ and the distribution $\mathsf{P}$ satisfies $(C_{\hc}, a)$-hypercontractivity if
    \begin{align*}
        \E \brac{\frac{1}{T} \sum_{t=1}^T \norm{g(X_t^1, U_t^1)}^4} \leq C_{\hc}\E \brac{\frac{1}{T} \sum_{t=1}^T \norm{g(X_t^1, U_t^1)}^2}^a \quad \forall g \in \calG,
    \end{align*}
    where the expectation is with respect to $\mathsf{P}$, the joint probability distribution of $\curly{X_t^1, U_t^1, X_{t+1}^1}_{t=1}^{T}$.
\end{definition}

We now state an exponential inequality which we will leverage along with hypercontractivity to bound the population risk by the empirical risk for any $g$ in the boundary of $\calG_\star^s$. 
\begin{lemma}
    \label{lem: mgf bound}
    Let $g$ belong to $\calG_\star^s$. For every $\lambda \in \R_+$ we have
    \begin{align*}
        \E \exp \paren{-\lambda \sum_{t,k=1}^{T,K} \norm{g(X_t^k,U_t^k)}^2} \leq \exp \paren{-\lambda K \E \brac{\sum_{t=1}^T \norm{g(X_t^1, U_t^1)}^2} + \frac{\lambda^2 K}{2} \E \brac{\paren{\sum_{t=1}^T \norm{g(X_t^1, U_t^1)}^2}^2}}.
    \end{align*}
\end{lemma}
\begin{proof}
    Observe that for $x \geq 0$, the inequality $e^{-x} \leq 1 - x + x^2/2 \leq e^{-x +x^2/2}$ holds. Using this and the fact that the trajectories are independent and identically distributed, we find
    \begin{align*}
         \E \exp \paren{-\lambda \sum_{t,k=1}^{T,K} \norm{g(X_t^k,U_t^k)}^2} &= \prod_{j\in[K]} \E \exp\paren{-\lambda \sum_{t=1}^T \norm{g(X_t^1,U_t^1)}^2} \\
        & \leq \prod_{j\in[K]} \E \paren{1 - \lambda \sum_{t=1}^T \norm{g(X_t^1,U_t^1)}^2 + \frac{\lambda^2}{2} \paren{\sum_{t=1}^T \norm{g(X_t^1, U_t^1)}}^2} \\
        &\leq \exp\paren{- \lambda K  \E\brac{ \sum_{t=1}^T\norm{g(X_t^1,U_t^1)}^2} + \frac{\lambda^2 K}{2} \E\paren{\sum_{t=1}^T \norm{g(X_t^1, U_t^1)}^2}^2}.
    \end{align*}
\end{proof}

Using the above exponential inequality along with a Chernoff bound provides the following high probability bound on the empirical prediction error by the population counterpart. 
\begin{lemma}
    \label{lem: single function lower iso}
    Suppose $(\partial B(r), \mathsf{P})$ is $(C_{\hc}, a)$-hypercontractive. For every $g \in \partial B(r)$, we have 
    \begin{align*}
        \P \brac{\frac{1}{TK}\sum_{t,k=1}^{T,K} \norm{g(X_t^k,U_t^k)}^2 < \frac{1}{2}\E\brac{\frac{1}{T} \norm{g(X_t^1,U_t^1)}^2} } &\leq \exp\paren{-\frac{K}{4 C_{\hc}} r^{4-2a}}.
    \end{align*}
\end{lemma}
\begin{proof}
    A Chernoff bound along with \Cref{lem: mgf bound} yields the estimate
    \begin{align*}
        &\P\brac{\frac{1}{TK} \sum_{t,k=1}^{T,K} \norm{g(X_t^k,U_t^k)}^2 < \frac{1}{2}\E\brac{\frac{1}{T} \norm{g(X_t^1,U_t^1)}^2}  } \\
        &\leq \inf_{\lambda \geq 0} \exp \paren{-\lambda K \E \brac{\sum_{t=1}^T \norm{g(X_t^1, U_t^1)}^2} + \frac{\lambda^2 K}{2} \E \brac{\paren{\sum_{t=1}^T \norm{g(X_t^1, U_t^1)}^2}^2}}\\
        &= \exp\paren{-\frac{K \E\brac{ \sum_{t=1}^T \norm{g(X_t^1,U_t^1)}^2}^2}{4 \E \brac{\paren{\sum_{t=1}^T \norm{g(X_t^1,U_t^1)}^2}^2}}}.
    \end{align*}
    By Cauchy-Schwarz,
    \begin{align*}
        \E \paren{\sum_{t=1}^T \norm{g(X_t^1,U_t^1)}^2}^2 \leq T \E\sum_{t=1}^T \norm{g(X_t^1,U_t^1)}^4 = T^2  \E\brac{\frac{1}{T}\sum_{t=1}^T \norm{g(X_t^1,U_t^1)}^4}. 
    \end{align*}
    Therefore, by the fact that $(\partial B(r), \pi)$ is $(C_{\hc},a)$-hypercontractive, we have that
    \begin{align*}
        \E \brac{\frac{1}{T}\sum_{t=1}^T \norm{g(X_t^1,U_t^1)}^4} \leq C_{\hc} \E \brac{\frac{1}{T}\sum_{t=1}^T \norm{g(X_t^1,U_t^1)}^2}^a.
    \end{align*}
    As a result, we may bound 
    \begin{align*}
        \exp\paren{-\frac{K \E\brac{ \sum_{t=1}^T \norm{g(X_t^1,U_t^1)}^2}^2}{4 \E \brac{\paren{\sum_{t=1}^T \norm{g(X_t^1,U_t^1)}^2}^2}}} \leq \exp\paren{-\frac{K \brac{\frac{1}{T}\sum_{t=1}^T \norm{g(X_t^1,U_t^1)}^2}^{2-a} }{4 C_{\hc}}}.
    \end{align*}
    To conclude, we note that $\brac{\frac{1}{T}\sum_{t=1}^T \norm{g(X_t,U_t)}^2} = r^2$ by the fact that $g \in \partial B(r)$. 
\end{proof}

The above bound for a single member of the function class may be extended to a uniform bound over function in $\calG_\star^s$, but outside the ball $B(r)$ using a covering argument.
\begin{lemma}[Lower isometry, Multi-trajectory Modification of Theorem 5.2 of \citet{ziemann2022learning}]
    \label{lem: lower iso} Suppose there exists a $r/\sqrt{8}$-net $\calG_r$ of $\partial B(r)$ in the norm $\norm{\cdot}_{\infty}$ such that $(\calG_r, \pi)$ satisfies $(C_{\hc}, a)$-hypercontractivity. Then the following lower isometry holds 
    \begin{align*}
        \P \brac{\forall g \in \calG_\star^s \setminus B(r), \frac{1}{TK} \sum_{t,k=1}^{T,K} \norm{g(X_t^k,U_t^k)}^2 - \frac{1}{8} \E\brac{\frac{1}{T} \norm{g(X_t^1, U_t^1)}^2} \leq 0} \leq \abs{\calG_r} \exp\paren{\frac{-K r^{4-2a}}{4 C_{\hc}}}.
    \end{align*}
\end{lemma}
\begin{proof}
   The fact that $\calG_\star^s$ is shaped allows us to the restrict attention to $g \in \partial B(r)$. By the fact that $\calG_r$ is an $r/\sqrt{8}$-net of $\partial B(r)$, we have by the reverse triangle inequality that for any $g \in \partial B(r)$, there exists some $g_i \in \calG_r$ such that 
   \begin{align}
        \label{eq: proj lower bound}
       \frac{1}{TK}\sum_{t,k=1}^{T,K} \norm{g(X_t^k,U_t^k)}^2 \geq \frac{1}{2TK} \sum_{t,k=1}^{T,K} \norm{g_i(X_t^k,U_t^k)}^2 - \frac{r^2}{8}. 
   \end{align}
   Define the event
   \begin{align*}
       \calE = \bigcup_{g \in \calG_r} \curly{\frac{1}{TK}\sum_{t,k=1}^{T,K} \norm{g(X_t^k,U_t^k)}^2 \leq \frac{1}{2} \E \brac{\frac{1}{T} \sum_{t=1}^T \norm{g(X_t^1,U_t^1)}^2}}.
   \end{align*}
    It follows from a union bound along with \Cref{lem: single function lower iso} that
    \begin{align*}
        \P\brac{\calE} \leq \abs{\calG_r} \exp\paren{\frac{-K r^{4-2a}}{4 C_{\hc}}}. 
    \end{align*}
    Fix an arbitrary $g \in \partial B(r)$. On the event $\calE^c$, we have
    \begin{align*}
        \frac{1}{TK}\sum_{t,k=1}^{T,K} \norm{g(X_t^k,U_t^k)}^2 &\geq \frac{1}{2TK}  \sum_{t,k=1}^{T,K} \norm{g_i(X_t^k,U_t^k)}^2 - \frac{r^2}{8} && \textrm{(We may find such a $g_i$ by \eqref{eq: proj lower bound})} \\
        &\geq \frac{1}{4} \E \brac{\frac{1}{T} \sum_{t=1}^T \norm{g_i(X_t^1,U_t^1)}^2} - \frac{r^2}{8} && \textrm{(By definition of $\calE$)} \\
        &= \frac{1}{8}\E \brac{\frac{1}{T} \sum_{t=1}^T \norm{g_i(X_t^1,U_t^1)}^2} && (g_i \in \partial B(r)).
    \end{align*}
    Since $g \in \partial B(r)$ was arbitrary, we have that
    \begin{align*}
        \P\brac{\sup_{g\in\partial B(r)} \frac{1}{TK} \sum_{t,k=1}^{T,K} \norm{g(X_t^k,U_t^k)}^2 - \frac{1}{8} \E \brac{\frac{1}{T}\norm{g(X_t^1,U_t^1)}^2} \leq 0} \leq  \abs{\calG_r} \exp\paren{\frac{-K r^{4-2a}}{4 C_{\hc}}}. 
    \end{align*}
    To extend the inequality in the expectation to hold for the set $\calG_\star^s \setminus B(r)$, we observe that by the fact that we may rescale $g$ to lie in $\partial B(r)$ since $\calG_\star^s$ is star-shaped.
\end{proof}

We now proceed with our proof of \Cref{lem: er bound}. 

\begin{proof}
We first verify that $\calG_\star^s$ is hypercontractive. 
By proposition 4.2 of \cite{ziemann2024sharp}, the fact that $\pi$ is a $(C_{\loja}, \alpha)$-Lojasiewicz policy and $\norm{D f(x,u;\phi)}  \leq L_f$ (Assumption~\ref{asmp: smooth dynamics}) implies that for all $g \in \calG_\star$, 
\[
    \lim_{p\to\infty} \paren{\E \brac{\frac{1}{T} \sum_{t=1}^T \norm{g(X_t^1, U_t^1)}^p}}^{1/p} \leq L_f C_{\loja} \alpha^{-1} \E \brac{\frac{1}{T} \sum_{t=1}^T \norm{g(X_t^1, U_t^1)}^2}^{\alpha/2}.
\]
We  additionally have that 
\[
     \paren{\E \brac{\frac{1}{T} \sum_{t=1}^T \norm{g(X_t^1, U_t^1)}^4}}^{1/4} \leq \lim_{p\to\infty} \paren{\E \brac{\frac{1}{T} \sum_{t=1}^T \norm{g(X_t^1, U_t^1)}^p}}^{1/p}.
\]
Combining these inequalities implies that $(\calG_\star, \mathsf{P})$ is $(C_{\hc}, a)$ hypercontractive, where $C_{\mathsf{hc}} = (L_f C_{\loja} \alpha^{-1})^4$ and $a = 4\alpha$.   By the following elementary inequality, the star-shaped hull $\calG_{\star}^s$ satisfies hypercontractivity with the same parameters: for $b \in [0,1]$ and $a 
 \leq 2$, $b^4 < b^{2a}$.  \\

\noindent \textbf{Bounding the population risk by the empirical risk: }

To bound the worst-case difference between the excess risk and its empirical counterpart over all functions in the class $\calG_\star$, we split the problem into two cases: one in which $g$ belongs to the ball $B(r)$, and one in which it does not.:
\begin{align*}
  &\sup_{g \in \calG_\star} \E \brac{\frac{1}{T} \sum_{t=1}^T \norm{g(X_t^1, U_t^1)}^2} - \frac{8}{TK}\sum_{t,k=1}^{T,K} \norm{g(X_t^k,U_t^k)}^2 
   \\
   &\leq \sup_{g \in \calG_\star \cap B(r)} \paren{\E \brac{\frac{1}{T} \sum_{t=1}^T \norm{g(X_t^1, U_t^1)}^2} - \frac{8}{TK}\sum_{t,k=1}^{T,K} \norm{g(X_t^k,U_t^k)}^2} \\
   &+ \sup_{g \in \calG_\star \setminus B(r)} \paren{ \E \brac{\frac{1}{T} \sum_{t=1}^T \norm{g(X_t^1, U_t^1)}^2} - \frac{8}{TK}\sum_{t,k=1}^{T,K} \norm{g(X_t^k,U_t^k)}^2}. 
\end{align*}
The first term is bounded by $r^2$ by the definition of $B(r)$ and the fact that 
$-\sum_{t,k=1}^{T,K} \norm{g(X_t^k,U_t^k)}^2 \leq 0.$ 
By \Cref{lem: lower iso}, the second term is less than zero with probability at least 
$1- \abs{\calG_r} \exp\paren{\frac{-K r^{4-2a}}{4 C_{\hc}}}$. Then under this event, the following bound holds for all $g \in \calG_\star$. 
\begin{align*}
     \E \brac{\frac{1}{T} \sum_{t=1}^T \norm{g(X_t^1, U_t^1)}^2} \leq 8 \frac{1}{TK} \sum_{t,k=1}^{T,K} \norm{g(X_t^k,U_t^k)}^2 + r^2.
\end{align*}

We may bound $\log \abs{\calG_r}$ using \eqref{eq: metric entropy of star hull} and \eqref{eq: metric entropy of parametric class} as
$
    \log \abs{\calG_r} \leq \log \paren{\frac{4B}{r}} + d_{\phi} \log \paren{\frac{3B L_f}{r }+1}. 
$ Then choosing $r^2 = \frac{1}{K^{a}}$, we find that by the requirement that
\[
    K \geq \paren{8 C_{\hc} \paren{ d_{\phi} \log\paren{4L_f K^{a/2}}} + \log \frac{1}{\delta} }^{\frac{1}{1-2a + a^2}}, 
\] 
the inequality 
\begin{align}
    \label{eq: lower iso}
     \E \brac{\frac{1}{T} \sum_{t=1}^T \norm{g(X_t^1, U_t^1)}^2} \leq 8 \frac{1}{TK} \sum_{t,k=1}^{T,K} \norm{g(X_t^k,U_t^k)}^2 + \frac{1}{K^a}
\end{align}
is satisfied for all $g \in \calG_\star$ with probability at least $1-\delta$. \\

\noindent \textbf{Bounding the empirical risk by the offset complexity: }
We now seek to bound the empirical risk. By the fact that $\hat \phi$ solves \eqref{eq: least squares}, it follows that
\begin{align*}
    \sum_{t,k=1}^{T,K} \norm{X_{t+1}^n - f(X_t^n, U_t^n; \hat\phi)}^2 \leq  \sum_{t,k=1}^{T,K} \norm{X_{t+1}^n - f(X_t^n, U_t^n; \phi^\star)}^2. 
\end{align*}
Substituting in $X_{t+1}^n = f(X_t^n, U_t^n; \phi^\star) + W_t^k$ and rearranging, we find 
\begin{align*}
    \sum_{t,k=1}^{T,K} \norm{f(X_t^n, U_t^n, \hat \phi) - f(X_t^n, U_t^n; \phi^\star)}^2 \leq 2 \sum_{t,k=1}^{T,K} \langle f(X_t^n, U_t^n; \hat \phi) - f(X_t^n, U_t^n; \phi^\star), W_t^k\rangle. 
\end{align*}
Multiplying the above inequality by two, and subtracting $\sum_{t,k=1}^{T,K} \norm{f(X_t^n, U_t^n, \hat \phi) - f(X_t^n, U_t^n; \phi^\star)}^2$ from each side leads to the following offset basic inequality:
\begin{align*}
    &\sum_{t,k=1}^{T,K} \norm{f(X_t^n, U_t^n, \hat \phi) - f(X_t^n, U_t^n; \phi^\star)}^2 \\
    &\leq 4 \sum_{t,k=1}^{T,K} \langle f(X_t^n, U_t^n; \hat \phi) - f(X_t^n, U_t^n; \phi^\star), W_t^k\rangle - \sum_{t,k=1}^{T,K} \norm{f(X_t^n, U_t^n, \hat \phi) - f(X_t^n, U_t^n; \phi^\star)}^2.
\end{align*}
For $g \in \calG_\star$, define the offset martingale complexity $M_{T,K}(g)$ as
\begin{align*}
    M_{T,K}(g) \triangleq \frac{1}{TK} \paren{4 \sum_{t,k=1}^{T,K} \langle g(X_t^n, U_t^n), W_t^k\rangle - \sum_{t,k=1}^{T,K} \norm{g(X_t^n, U_t^n)}^2}.
\end{align*}
Then the offset basic inequality allows us to bound the empirical risk by the offset martingale complexity as 
\begin{align}
\label{eq: bound empirical by offset}
\frac{1}{TK}\sum_{t,k=1}^{T,K} \norm{g(X_t^k,U_t^k)}^2 \leq M_{T,K}(g).
\end{align}
    
\noindent \textbf{Bounding the offset complexity:}
We bound the offset complexity using a covering argument. 

Let $\calG_{\gamma}$ be a minimal cardinality $\gamma$-cover for $\calG_\star$ with cardinality $\calN(\calG_\star, \norm{\cdot}_{\infty}, \gamma)$. Also let $\Pi$ denote the projection of an element $g \in \calG_\star$ onto $\calG_{\gamma}$. We have that, 
\begin{align*}
     \sup_{g \in \calG_{\star}} M_{T,K}(g) &= \sup_{g \in \calG_{\star}} \frac{1}{TK} \paren{ \sum_{t,k=1}^{T,K} \paren{4\langle g(X_t^n, U_t^n), W_t^k\rangle - \norm{g(X_t^n, U_t^n)}^2}} \\
     &= \sup_{g \in \calG_{\star}}\frac{1}{TK} \paren{ \sum_{t,k=1}^{T,K} \paren{4\langle g (X_t^n, U_t^n)-g_{\Pi g}(X_t^n, U_t^n), W_t^k\rangle + 4\langle g_{\Pi g}(X_t^n, U_t^n), W_t^k\rangle - \norm{g(X_t^n, U_t^n)}^2}}.
\end{align*}
For any $g \in \calG_\star$, we may apply the Cauchy-Schwarz inequality and the bound $\norm{g - \Pi g}_{\infty} \leq \gamma$ to bound
\begin{align*}
    \sup_{g \in \calG_{\star}} \frac{1}{TK} \sum_{t,k=1}^{T,K} 4\langle g(X_t^n, U_t^n) - g_{\Pi g}(X_t^n, U_t^n), W_t^k \rangle &\leq 4 \gamma \sqrt{\frac{\sum_{t,k=1}^{T,K} \norm{W_t^k}^2}{TK}} \leq 12 \gamma \sigma_w \dx 
\end{align*}
with probability at least $1-\delta$ using the concentration of chi-squared random variables, and the fact that $K \geq 3 \log\paren{1/\delta}$. Additionally note that 
\begin{align*}
    \norm{g(X_t^n, U_t^n)}^2 \geq \frac{1}{2} \norm{g_{\Pi g}(X_t^n, U_t^n)}^2 - \norm{g(X_t^n, U_t^n) - g_{\Pi g}(X_t^n, U_t^n) }^2 \geq \frac{1}{2} \norm{g_{\Pi g}(X_t^n, U_t^n)}^2 - \gamma^2.
\end{align*}
Combining the above inequalities allows us to bound the offset complexity over all parameters in terms of the offset complexity of the finite set of parameters determined by the cover $\calG_\gamma$.
\begin{align}
    \label{eq: original offset complexity bound}
    \sup_{g \in \calG_{\star}} M_{T,K}(g) \leq 12 \gamma \sigma_w \dx  + \gamma^2 + \frac{\sup_{g \in \calG_{\gamma}} \frac{1}{TK} \paren{ \sum_{t,k=1}^{T,K} \paren{4\langle g(X_t^n, U_t^n), W_t^k\rangle - \frac{1}{2}\norm{g(X_t^n, U_t^n)}^2}}}{TK}.
\end{align}
We may then in turn bound the offset complexity over a finite hypothesis class using Proposition F.2 of \citet{ziemann2023tutorial}. With probability at least $1-\delta$, 
\begin{equation}
\begin{aligned}
    \label{eq: finite offset complexity bound}
    &\sup_{g \in \calG_\gamma} \frac{1}{TK} \paren{ \sum_{t,k=1}^{T,K} \paren{4\langle g(X_t^n, U_t^n), W_t^k\rangle - \frac{1}{2}\norm{g(X_t^n, U_t^n)}^2}} \\&\leq 16 \sigma_w^2 \log\paren{\frac{\abs{\calG_\gamma}}{\delta}} \leq 16 \sigma_w^2 \paren{d_{\phi} \log\paren{\frac{3 B L_f}{\gamma}+1} + \log \frac{1}{\delta}}, 
\end{aligned}
\end{equation}
where the final inequality followed from \eqref{eq: metric entropy of parametric class}. Setting $\gamma = \frac{\sigma_w}{TK}$, we may combine \eqref{eq: original offset complexity bound} with \eqref{eq: finite offset complexity bound} to conclude 
\begin{align*}
    \sup_{g \in \calG_{\star}} M_{T,K}(g) &\leq \frac{12  \sigma_w^2 \dx}{TK}  + \frac{\sigma_w^2}{(TK)^2} + 16 \sigma_w^2 \paren{d_{\phi} \log\paren{3 B L_fTK/\sigma_w+1} + \log \frac{1}{\delta}} \\ &\leq \frac{32 \sigma_w^2}{TK} \paren{\dx + d_{\phi}\log\paren{3BL_f TK/\sigma_w + 1} + \log\frac{1}{\delta}}.
\end{align*}
By the burn-in requirement 
\[
    K \geq \paren{\frac{256 \sigma_w^2}{T} \paren{\dx + d_{\phi}\log\paren{4L_f TK/\sigma_w} + \log\frac{1}{\delta}}}^{\frac{1}{a-1}},
\]
we may combine the above inequality with \eqref{eq: lower iso} and \eqref{eq: bound empirical by offset} to achieve the desired result. 
\end{proof}

With this result in hand, we may immediately determine a high probability bound on the parameter estimation error as it appears in the bound in \Cref{lem: cost decomposition}. 
\begin{corollary}
    \label{cor: slow rate}
    Under the setting of \Cref{lem: er bound}, it follows that with probability at least $1-\delta$,
    \begin{align*}
        &\norm{\hat \phi \!  -\!\phi^\star}^2\!\leq\! \paren{\frac{512 \sigma_w^2}{T K } \paren{\dx \!+\!d_{\phi}\log\paren{4 B L_f T K}} + \log \frac{1}{\delta}}^{2\alpha}.
    \end{align*}
\end{corollary}
For $\alpha < \frac{1}{2}$, the above bound decays at a sub-optimal rate. However, the dependence can be improved by leveraging the Delta method, outlined in the subsequent section.

\subsection{Improved Rates via the Delta Method}\label{s: delta method}
We can improve upon the estimation error bounds from the previous section by leveraging the smoothness assumptions on the dynamics. In particular, the loss function is locally quadratic in the parameter error once when the parameter error is sufficiently small. Therefore, once we have sufficient data available, the parameter error should scale with the rate of $\frac{1}{TK}$. 

We will show that this is the case using an appeal to Taylor's theorem. In particular, we apply Taylor's theorem to the optimality conditions of the least squares problem, an approach known as the delta method \citep{van2000asymptotic}. Doing so allows us to to show that the parameter error may be expressed in terms of a self-normalized martingale. To express this result, we define the shorthand $D f_t^k \triangleq D_{\phi} f(X_t^k, U_t^k; \phi^\star)$. In differentiating the optimality conditions, we encounter an empirical covariance for the covariates defined by the Jacobian of $f$ evaluated at the experiment data:
\begin{align}
\label{eq: emp covariance def}
    \hat { \Sigma}_K^\pi = \frac{1}{K} \sum_{k=1}^K \sum_{t=1}^T (D f_t^k)^\top D f_t^k.
\end{align}
The superscript $\pi$ denotes the policy that is rolled out to generate $U_t^k$. Given sufficiently many trajectories, this estimate converges to the expected value \begin{align}\label{eq: covariance} \Sigma^\pi \triangleq \mathbb{E}_\pi^{\phi^\star}\brac{\sum_{t=1}^T D f(X_t, U_t, \phi^\star)^\top D f(X_t, U_t, \phi^\star)}.\end{align} 
As long as this covariance matrix is positive definite, $\Sigma^{\pi} \succ 0$, we can show that the delta method is successful in improving the estimation rate. 

In order to do so, we require sufficiently many trajectories. In addition to the burn-in for the excess risk bound, we require that $K$ is large enough that \Cref{cor: slow rate} guarantees  $\norm{\hat \phi_K - \phi^\star} \leq 1$ with high probability. We refer to the burn-in required for this to hold as $K \geq \tau_{\mathsf{small\,error}}(\delta)$, where 
\begin{align*}
     \tau_{\mathsf{small\,error}}(\delta) = \frac{512 \sigma_w^2}{T} \paren{\dx + d_{\phi}\log\paren{4B L_f TK/\sigma_w} + \log\frac{2}{\delta}}.
\end{align*}

Furthermore, we require that the higher order terms ($\mathsf{h.o.t.}$) from the taylor expansion involved in the Delta method are dominated by the lower order terms. For downstream use of the delta method, we allow this to depend on the norm of an arbitrary positive definite matrix $H$ and an arbitrary positive number $\beta$. We refer to the required burn-in as $K \geq \tau_{\mathsf{h.o.t.}}(\delta, \beta)$, where
\begin{align*}
    \tau_{\mathsf{h.o.t}}(\delta, H) &= \Bigg(\frac{8T^{1-\alpha} L_f^2 }{\beta }  \paren{512 \sigma_w^2\paren{\dx + d_{\phi}\log\paren{4 BL_f TK/\sigma_w } + \log\frac{2}{\delta}}}^\alpha  + \frac{8 \dx \sigma_w L_f }{ \beta} \sqrt{32 T \log\frac{4}{\delta}} \Bigg)^{\frac1\alpha}.
\end{align*}

With these burn-in times defined, we may state the result which expresses the parameter estimation error 

\begin{restatable}{lemma}{deltamethod}
\label{lem:deltamethod}
     Let Assumption~\ref{asmp: smooth dynamics} hold and fix $\delta \in (0,\frac{1}{2}]$. Let $H$ be a positive definite matrix and $\beta$ a positive number. Fix some $\delta \in (0,\frac{1}{4}]$.  Consider setting $\hat \phi$ as the least squares solution from $K$ episodes of data collected with policy $\pi$ for 
    \begin{align*}
        K &\geq \max\curly{ \tau_{\ER}(\delta/2), \tau_{\mathsf{small\,error}}(\delta), \tau_{\mathsf{h.o.t}}(\delta, \beta)} 
    \end{align*}
    Suppose $\hat \Sigma_K^\pi \succeq \frac{1}{2}\lambda_{\min} (\Sigma^\pi) > 0$.  Then with probability at least $1-\delta$, 
    \begin{align*}
        \hat \phi - \phi^\star &= (I+G)(\hat \Sigma_K^\pi)^{-1} \frac{1}{K} \sum_{t,k=1}^{T,K} (D f_t^k)^\top W_t^k,
    \end{align*}
    where $G$ is a matrix satisfying $\norm{G}_{\op} \leq \beta$. 
\end{restatable}

As long as $\Sigma^\pi \succ0$ and $K$ is sufficiently large that $\Sigma_K^\pi$ concentrates to $\Sigma^\pi$,  we may use the above result to show that $\hat \phi_K - \phi^\star$ behaves like a self-normalized martingale. We may then  bound this self-normalized martingale, leading to \Cref{thm: identification error bound}. Before proving the above lemma, we introduce several preliminary results.
\begin{lemma}
    \label{lem: woodbury bound}
    Let $A$ be a symmetric, positive definite matrix. Let $B$ be such that $\norm{B} \leq s \lambda_{\min}(A)$ for $s \leq \frac{1}{2}$. Then 
   $
        (A+B)^{-1} C = (I + G) A^{-1} C$
    for some $G$ satisfying $\norm{G} \leq 2s.$
\end{lemma}
\begin{proof}
    Let $B = U\Sigma V^\top$ be the thin singular value decomposition of $B$. By the woodbury matrix idenity
    \begin{align*}
        (A+U\Sigma V^\top)^{-1} C = \paren{I - A^{-1} U ( \Sigma^{-1} + V^\top A^{-1} U)^{-1} V } A^{-1} C.
    \end{align*}
    Let $G = - A^{-1} U ( \Sigma^{-1} + V^\top A^{-1} U)^{-1} V$. By submult iplicativity, $
        \norm{G}_{\op} \leq \frac{1}{\lambda_{\min}(A)} \paren{\sigma_{\min}(\Sigma^{-1} + V^\top A^{-1} U)}^{-1}.$
    We have that
    \begin{align*}
        \sigma_{\min}(\Sigma^{-1} + V^\top A^{-1} U) &\geq \sigma_{\min}(\Sigma^{-1}) - \norm{V^\top A^{-1} U} = \sigma_{\min}(\Sigma^{-1}) - \frac{1}{\lambda_{\min}(A)}. 
    \end{align*}
    Then by noting that $\sigma_{\min}(\Sigma^{-1}) = \norm{B}$, we have 
    $\norm{G}_{\op} \leq  \paren{\frac{\lambda_{\min}(A)}{\norm{B}} -1}^{-1}$
    The result follows by substituting the bound on $\norm{B}$. 
\end{proof}

\begin{lemma}[Martingale concentration]
    \label{lem: martingale concentration}
    Suppose $S_1, \dots, S_T$ is a real-valued martingale process adapted to the filtration $\calF_1, \dots, \calF_T$ with conditionally $\sigma^2$-sub-Gaussian increments. Then 
    \begin{align*}
        \P\brac{S_T \geq \rho} \leq \exp\paren{-\frac{\rho^2}{2 T \sigma^2}}.
    \end{align*}
\end{lemma}
\begin{proof}
    By a Chernoff bound, 
    $
        \P \brac{S_T \geq \rho} \leq \inf_{\lambda \geq 0} \exp(-\rho \lambda) \E \brac{\exp\paren{\lambda S_T}}.$
    From the fact that the increments are sub-Gaussian, we have
    \begin{align*}
        \E \brac{\exp\paren{\lambda S_T}} &= \E \E  \brac{\exp\paren{\lambda S_T} \vert \calF_{T-1}}= \E\brac{\exp\paren{\lambda S_{T-1}} \E \brac{\exp\paren{\lambda (S_T - S_{T-1})} \vert \calF_{T-1} }} \\
        &\leq \E\brac{\exp\paren{\lambda S_{T-1}} }\exp\paren{\frac{\sigma^2 \lambda^2}{2}} \leq \exp\paren{T \frac{\sigma^2 \lambda^2}{2}}. \quad (\mbox{Increments are conditionally sub-Gaussian})
    \end{align*}
    Substituting into the Chernoff bound and optimizing over $\lambda$ provides the result. 
\end{proof}

We now proceed to prove \Cref{lem:deltamethod}
\begin{proof}
By the fact that $K\geq \tau_{\ER}(\delta/2)$, the conditions on \Cref{cor: slow rate} are satisfied such that with probability at least $1-\delta/2$, 
\begin{align}
    \label{eq: delta proof consistency}
    \norm{\hat \phi_K  -\phi^\star} \leq \paren{\frac{512 \sigma_w^2}{TK } \paren{\dx + d_{\phi}\log\paren{3L_f TK/\sigma_w} + \log\frac{2}{\delta}}}^{\alpha}.
\end{align}
Using the assumption $K \geq \tau_{\mathsf{small\,error}}(\delta)$,   
we have that this is at most $1$.
Under this event, and by the assumption that $\norm{\phi^\star} \leq B-1$, we have that $\hat \phi_K$ is is in the interior of $\curly{\phi: \norm{\phi}\leq B}$. Thus by the fact that $\hat\phi_K$ solves \eqref{eq: least squares}, it holds that   
\begin{align*}
    0 & =\paren{\nabla_{\phi} \sum_{k=1}^K \sum_{t=1}^T \norm{X_{t+1}^k - f(X_t^k, U_t^k; \phi)}^2} \vert_{\phi = \hat \phi_K} \triangleq \psi(\hat \phi_K). 
\end{align*}
A first order Taylor expansion of $\psi(\hat \phi_K)$ about $\phi^\star$ provides
\begin{align*}
    \psi(\hat \phi_K) &= \psi(\phi^\star) + D_\phi \psi(\phi)\vert_{\phi = \phi_\star} (\hat \phi_K - \phi^\star) + \mathsf{Remainder}_{\psi}(\hat \phi_K)[\cdot, \hat\phi_K - \phi_\star, \hat\phi_K- \phi_\star],
\end{align*} 
where $\mathsf{Remainder}_{\psi}(\hat\phi_K)$ is a trilinear operator such that the vector $\mathsf{Remainder}_{\psi}(\hat \phi_K)[\cdot, \hat\phi_K - \phi_\star, \hat\phi_K- \phi_\star]$ is defined by the product of the matrix $\mathsf{Remainder}_{\psi}(\hat \phi_K)[\cdot, \cdot \hat\phi_K- \phi_\star]$ with $\hat \phi_K - \phi_\star$: $\mathsf{Remainder}_{\psi}(\hat \phi_K)[\cdot, \hat\phi_K - \phi_\star, \hat\phi_K- \phi_\star] = \mathsf{Remainder}_{\psi}(\hat \phi_K)[\cdot, \cdot \hat\phi_K- \phi_\star] (\hat\phi_K - \phi_\star)$, and the satisfies the operator norm bound below:  
\begin{align*}
    \norm{\mathsf{Remainder}_{\psi}(\hat \phi_K) [\cdot, \cdot, \hat\phi_K - \phi_\star]}_{\mathsf{op}} &\leq \sup_{\ell \in [0,1]} \norm{D_{\phi}^2 \psi(\phi)\vert_{\phi =\ell \phi^\star + (1-\ell) \hat \phi_K}[\cdot, \cdot, \hat \phi_K - \phi^\star]}_{\mathsf{op}}.
\end{align*}

As long as $D_{\phi} \psi(\phi)\vert_{\phi = \phi^\star} +\mathsf{Remainder}_{\psi}(\hat \phi_K)[\cdot, \cdot, \hat\phi_K - \phi_\star]$ is non-singular, then the above two equations imply
\begin{align*}
    \hat \phi_K &- \phi^\star = \paren{D_{\phi} \psi(\phi)\vert_{\phi = \phi^\star} + \mathsf{Remainder}_{\psi}(\hat \phi_K)[\cdot, \cdot, \hat\phi_K - \phi_\star]}^{-1} \psi(\phi^\star).
\end{align*}
Evaluating the derivatives provides
\begin{align*}
    \psi(\phi^\star) &= \sum_{k=1}^K \sum_{t=1}^T (D f_t^k)^\top W_t^k \\
    D_{\phi} \psi(\phi)\vert_{\phi = \phi^\star} &= K \hat \Sigma_K^\pi + \underbrace{\sum_{k=1}^K \sum_{t=1}^T (D^2 f_t^k)[\cdot, \cdot, W_t^k]}_{E_1}, \\
    D_{\phi}^2 \psi(\phi)\vert_{\phi = \tilde \phi}[\cdot, \cdot, \hat \phi_K - \phi^\star] &= \underbrace{\sum_{k=1}^K \sum_{t=1}^T D^2 f(X_t^k, U_t^k; \tilde \phi)[\cdot, \cdot, \hat\phi_K - \phi_\star]^\top D f (X_t^k, U_t^k; \tilde \phi)}_{E_2 
    } \\
    &+ \underbrace{\sum_{k=1}^K \sum_{t=1}^T (D f(X_t^k, U_t^k; \tilde\phi )^\top D^2 f(X_t^k, U_t^k; \tilde \phi)[\cdot,\cdot, \hat \phi_K - \phi^\star]}_{E_3} \\
    & +  \underbrace{\sum_{k=1}^K \sum_{t=1}^T (D^3 f(X_t^k, U_t^k; \tilde \phi))[\cdot, \cdot, W_t^k, \hat \phi_K -\phi^\star]}_{E_4}
\end{align*}
We will bound the norms of $E_1, E_2, E_3$ and $E_4$ for $\tilde\phi$ an arbitrary convex combination of $\hat\phi_K$ and $\phi_\star$ in order to apply \Cref{lem: woodbury bound}. We may bound $E_2$ and $E_3$ directly using Assumption~\ref{asmp: smooth dynamics}:
\begin{align*}
    \norm{E_2}_{\op} &\leq TK L_f^2 \norm{\hat \phi_K -\phi_\star} \mbox{ and }  \norm{E_3}_{\op} \leq TK L_f^2 \norm{\hat \phi_K -\phi_\star}.
\end{align*}
For $E_1$ and $E_4$, we use a concentration argument. First note that by expressing the $i^{th}$ entry of $W_t^k$ as $W_t^k[i]$ and the $i^{th}$ entry of $f_t^k$ as $f_t^k[i]$, we have 
\begin{align*}
    E_1 = \sum_{i=1}^{\dx} \sum_{k=1}^K \sum_{t=1}^T W_t^k[i] D^2 (f_t^k[i])
\end{align*}
Appealing to a covering argument (cf. Lemma 2.4 of \cite{ziemann2023tutorial}), we may construct a $\frac{1}{2}$-net $\calN$ for the unit sphere in $\R^{\dx}$. Then we have $\P \brac{\norm{E_1}_{\op} \geq \rho} \leq 5^{\dx^2} \sup_{u,v \in \calN} \P \brac{u^\top E_1 v \geq \frac{1}{4} \rho}$. 
For any $u,v\in\R^{\dx}$,
\begin{align*}
    u^\top E_1 v =  \sum_{i=1}^{\dx} \sum_{k=1}^K \sum_{t=1}^T W_t^k[i] u^\top D^2 (f_t^k[i]) v.
\end{align*}
The using boundedness, Assumption~\ref{asmp: smooth dynamics}, the quantity $\sum_{t=1}^T W_t^k[i] u^\top D^2 (f_t^k[i]) v$ is a martingale with conditionally $\sigma_w^2 L_f^2$-sub-Gaussian increments. Then by \Cref{lem: martingale concentration}, 
\begin{align*}
    \P\brac{\sum_{t=1}^T W_t^k[i] u^\top D^2 (f_t^k[i]) v \geq \frac{\rho}{4}} \leq \exp\paren{-\frac{\rho^2}{32 TK \sigma_w^2 L_f^2}}. 
\end{align*}
Using this bound in the covering argument above, we find
\begin{align*}
    \P \brac{\norm{E_1}_{\op} \geq \rho} \leq 5^{\dx^2}\exp\paren{-\frac{\rho^2}{32 TK \sigma_w^2 L_f^2}}.
\end{align*}
As a result, with probability at least $1-\delta/4$, 
\begin{align*}
    \norm{E_1}_{\op} \leq \dx \sigma_w L_f \sqrt{ 32 TK \log\frac{4}{\delta}}.  
\end{align*}
Similarly, we may show that with probability at least $1-\delta/4$
\begin{align*}
    \norm{E_4}_{\op} \leq \dx \sigma_w L_f \norm{\hat \phi_K - \phi^\star} \sqrt{32 TK \log\frac{4}{\delta}}.  
\end{align*}
Union bounding, we find that with probability at least $1-\delta/2$, 
\begin{align*}
    \norm{E_1 + E_2 + E_3 + E_4}_{\op} \leq \dx \sigma_w L_f (1 + \norm{\hat \phi_K - \phi^\star}) \sqrt{32 TK \log\frac{4}{\delta}} + 2 TK L_f^2 \norm{\hat \phi_K -\phi_\star}. 
\end{align*}
Using the fact that $\norm{\hat \phi_K - \phi^\star} \leq 1$ along with the bound in \eqref{eq: delta proof consistency} it holds that with probability at least $1-\delta$, 
\begin{align*}
    \norm{E_1 + E_2 + E_3 + E_4}_{\op} \leq 2 \dx \sigma_w L_f \sqrt{32 TK \log\frac{4}{\delta}} + 2 (TK)^{1-\alpha} L_f^2 \paren{ 512 \sigma_w^2\paren{\dx + d_{\phi}\log\paren{4L_f TK/\sigma_w } + \log\frac{2}{\delta}}}^{\alpha}
\end{align*}
As $K \geq \tau_{\mathsf{h.o.t}}(\delta,  \beta)$ and $\lambda_{\min} (\hat \Sigma_K^\pi) \geq \frac{1}{2} \lambda_{\min}(\Sigma^\pi)$, it holds  that
\begin{align*}
    \norm{E_1 + E_2 + E_3 + E_4}_{\op} \leq \frac{\beta}{2} 
    K \lambda_{\min}\paren{\hat \Sigma_K^\pi}. 
\end{align*}
We may therefore invoke \Cref{lem: woodbury bound} with $s = \frac{\beta}{2}$ 
to show that 
\begin{align*}
    &\paren{D_{\phi} \psi(\phi)\vert_{\phi = \phi^\star} + \mathsf{Remainder}_{\psi}(\hat\phi_K)[\cdot, \cdot, \hat\phi_K - \phi_\star]}^{-1}  \psi(\phi^\star) = (I + G) (\hat \Sigma_K^\pi)^{-1} \frac{1}{K} \sum_{t,k=1}^{T,K} (D f_t^n)^\top W_t^k,
\end{align*}
for some $G$ satisfying $\norm{G}_{\op} \leq \beta$. 

\end{proof}

\subsection{Proof of \Cref{thm: identification error bound}}
\label{s: proof of identification error bound}

To proceed with our proof of the main result, we will leverage \Cref{lem:deltamethod}. In particular, we require a bound on the weighted norm of a self-normalized martingale. 

\begin{lemma}
\label{lem: sn martingale bnd}
    Suppose $\Sigma^\pi \succ 0$ and $H \succ 0$. Define the event $\calE \triangleq \curly{\norm{\hat \Sigma_K^\pi - \Sigma^\pi} \leq \beta}$ for some positive $\beta$ satisfying 
    \[
        \beta \leq \frac{\lambda_{\min}(\Sigma^\pi)}{2} 
    \]
    As long as  $\calE$ holds, the following holds with probability at least $1-\delta$, 
    \begin{align*}
        &\norm{\paren{K \hat\Sigma_K^\pi}^{-1}\sum_{t,k=1}^{T,K} (D f_t^n)^\top W_t^k}_H^2 \leq 2 \sigma_w^2 \paren{1+\frac{4\beta}{\lambda_{\min}(\Sigma^\pi)}}\trace( \paren{K \Sigma^\pi}^{-1} H) + 4 \sigma_w^2 \paren{1+\frac{4\beta}{\lambda_{\min}(\Sigma^\pi)}}\norm{(K \Sigma^\pi)^{-1} H} \log\frac{1}{\delta}.
    \end{align*}
\end{lemma}



    \begin{proof}
    Define 
    \[
        \boldsymbol Z = \bmat{D f_1^1 \\ \vdots \\D f_T^1 \\ \vdots \\D f_1^{K} \\ \vdots \\ D f_T^K}  \quad\mbox{ and }\quad \boldsymbol w = \bmat{W_1^1  \\ \vdots \\W_t^1 \\ \vdots \\ W_1^K \\ \vdots \\W_t^K}.
    \]
    such that 
    \begin{align*}
        \norm{\paren{K \hat\Sigma_K^\pi}^{-1}\sum_{t,k=1}^{T,K} (D f_t^n)^\top W_t^k}_H^2 = \norm{(\bZ \bZ^\top)^{-1} \bZ \bW }_H^2 \mbox{ and } \hat \Sigma_K^\pi =\frac{1}{K} \bZ \bZ^\top. 
    \end{align*}
    We may write 
    \begin{align*}
        \norm{\paren{\bZ\bZ^\top}^{-1} \bZ \bW }_H^2 &= \norm{H^{1/2}\paren{\bZ\bZ^\top}^{-1} \bZ \bW }^2 \\
        &= \norm{\paren{H^{-1/2}\bZ\bZ^\top H^{-1/2}}^{-1} H^{-1/2}\bZ \bW }^2 \\
        &= \norm{\paren{\paren{H^{-1/2}\bZ\bZ^\top H^{-1/2}}^{2}}^{-1/2} H^{-1/2}\bZ \bW }^2\\
        &= 2\norm{\paren{\paren{H^{-1/2}\bZ\bZ^\top H^{-1/2}}^2  + \paren{H^{-1/2}\bZ\bZ^\top H^{-1/2}}^{2}}^{-1/2} H^{-1/2}\bZ \bW }^2
    \end{align*}
    We have by assumption that $\norm{\frac{1}{K}\bZ \bZ^\top - \Sigma^\pi} \leq \beta$, and therefore that \[H^{-1/2} \bZ\bZ^\top H^{-1/2} \succeq \paren{1- \frac{\beta}{\lambda_{\min}(\Sigma^\pi)}} H^{-1/2} K \Sigma^{\pi} H^{-1/2}.\]
    It follows from the above inequality and the assumption that $\beta \leq \frac{\lambda_{\min}(\Sigma^\pi)}{2} $ that
    \begin{align*}
        &2\norm{\paren{\paren{H^{-1/2}\bZ\bZ^\top H^{-1/2}}^2  + \paren{H^{-1/2}\bZ\bZ^\top H^{-1/2}}^{2}}^{-1/2} H^{-1/2}\bZ \bW }^2 \\
        &\leq 2 \paren{1 + \beta \frac{2}{ \lambda_{\min}(\Sigma^\pi)}}\norm{\paren{\paren{H^{-1/2} K\Sigma^\pi H^{-1/2}}^2  + \lambda_{\min}\paren{H^{-1/2} K\Sigma^\pi H^{-1/2}} H^{-1/2}\bZ\bZ^\top H^{-1/2}}^{-1/2} H^{-1/2}\bZ \bW }^2 \\
        &= 2 \paren{1 + \beta \frac{2}{\lambda_{\min}(\Sigma^\pi)}} \lambda_{\min}\paren{H^{-1/2} K\Sigma^\pi H^{-1/2}}^{-1} \norm{\paren{\frac{\paren{H^{-1/2} K\Sigma^\pi H^{-1/2}}^2}{\lambda_{\min}\paren{H^{-1/2} K \Sigma^\pi H^{-1/2}}}  +  H^{-1/2}\bZ\bZ^\top H^{-1/2}}^{-1/2} H^{-1/2}\bZ \bW }^2.
    \end{align*}
    The above quantity is a self-normalized martingale. In particular, we can bound it by applying the bound in Theorem 4.1 of \citet{ziemann2023tutorial}. To do so, define $\tilde \Sigma \triangleq \frac{\paren{H^{-1/2} K\Sigma^\pi H^{-1/2}}^2}{\lambda_{\min}\paren{H^{-1/2} K \Sigma^\pi H^{-1/2}}}$. Then it follows that with probability at least $1-\delta$ that
    \begin{align*}
        &2 \paren{1 + \beta \frac{2}{ \lambda_{\min}(\Sigma^\pi)}} \lambda_{\min}\paren{H^{-1/2} K \Sigma^\pi H^{-1/2}}^{-1} \norm{\paren{\tilde \Sigma  +  H^{-1/2}\bZ\bZ^\top H^{-1/2}}^{-1/2} H^{-1/2}\bZ \bW }^2 \\
        &\leq 2\paren{1 + \beta \frac{2}{\lambda_{\min}(\Sigma^\pi)}} \lambda_{\min}\paren{H^{-1/2} K \Sigma^\pi H^{-1/2}}^{-1}\paren{\sigma_w^2 \log\det\paren{I + H^{-1/2} \bZ \bZ^\top H^{-1/2} \tilde \Sigma^{-1}} + 2 \sigma_w^2 \log\frac{1}{\delta}} \\
        &\leq 2\paren{1 + \beta \frac{2}{\lambda_{\min}(\Sigma^\pi)}} \lambda_{\min}\paren{H^{-1/2} K \Sigma^\pi H^{-1/2}}^{-1}\paren{ \sigma_w^2 \trace\paren{H^{-1/2} \bZ \bZ^\top H^{-1/2} \tilde \Sigma^{-1}} + 2 \sigma_w^2 \log\frac{1}{\delta}},
    \end{align*}
    where the second inequality follows by noting that for any positive definite matrix  matrix $M \in \R^{d\times d}$, $\log\det(I+M) = \sum_{i=1}^d \log (1 + \lambda_i(M)) \leq \sum_{i=1}^d \lambda_i(M)$.
    By the concentration event on $\bZ \bZ^\top$ and the definition of $\tilde \Sigma$, this quantity may be bounded as 
    \begin{align*}
        &2\paren{1 + \beta \frac{2}{\lambda_{\min}(\Sigma^\pi)}} \lambda_{\min}\paren{H^{-1/2} K \Sigma^\pi H^{-1/2}}^{-1}\paren{\sigma_w^2 \trace\paren{H^{-1/2} \bZ \bZ^\top H^{-1/2} \Sigma^{-1}} + 2 \sigma_w^2 \log\frac{1}{\delta}} \\
        &\leq 2\paren{1 + \beta \frac{2}{ \lambda_{\min}(\Sigma^\pi)}} \lambda_{\min}\paren{H^{-1/2} K \Sigma^\pi H^{-1/2}}^{-1} \paren{\sigma_w^2 \paren{1+\frac{\beta}{\lambda_{\min}(\Sigma^\pi)}}\trace\paren{H^{-1/2} K\Sigma^\pi H^{-1/2} \tilde \Sigma^{-1}} + 2 \sigma_w^2 \log\frac{1}{\delta}} \\
        &=  2\paren{1 + \beta \frac{2}{ \lambda_{\min}(\Sigma^\pi)}} \paren{\sigma_w^2 \paren{1+\frac{\beta}{ \lambda_{\min}(\Sigma^\pi)}}  \trace\paren{H \paren{K \Sigma^\pi}^{-1}} + 2 \sigma_w^2 \norm{H \paren{K\Sigma^\pi}^{-1}} \log\frac{1}{\delta}}\\
        &\leq 2 \paren{1 + \beta \frac{4}{\lambda_{\min}(\Sigma^\pi)}}\paren{\sigma_w^2  \trace\paren{H \paren{\Sigma^\pi}^{-1}} + 2 \sigma_w^2 \norm{H \paren{\Sigma^\pi}^{-1}} \log\frac{1}{\delta}},
    \end{align*}
    where the final inequality follows from the fact that $\beta \leq \frac{\lambda_{\min}(\Sigma^\pi)}{2}$. This concludes the proof.

\end{proof}

To guarantee concentration of the empirical covariance \eqref{eq: emp covariance def} to the true covariance, we state the following covariance concentration result. 
\begin{lemma}
    \label{lem: covariance concentration}
    Let $\beta \in (0,\frac{\lambda_{\min}(\Sigma^{\pi})}{4})$. Suppose 
    \begin{align*}
        K \geq \frac{2 T^2 L_f^4 (3 d_{\phi} + \log(1/\delta))}{\beta^2 }. 
    \end{align*} Then with probability at least $1-\delta$, 
    \begin{align*}
        \norm{\hat \Sigma_K^{\pi} - \Sigma^\pi} \leq \beta. 
    \end{align*}
\end{lemma}
\begin{proof}
    Let $\calN$ be an $\varepsilon$-net of $\curly{v \in \R^{d_\phi}: \norm{v}=1}$ with $\varepsilon=\frac{1}{4}$. By a covering argument (Lemma 2.5 of \cite{ziemann2023tutorial}), 
    \begin{align*}
        \P\brac{\norm{\hat \Sigma_K^\pi - \Sigma^\pi} \geq \rho} \leq 9^{d_{\phi}} \max_{v \in \calN} \P\brac{\abs{v^\top (\hat \Sigma_K^\pi - \Sigma^\pi) v \geq \frac{\rho}{2}}}.
    \end{align*}
     Let $\hat \Sigma^{\pi,k} = \sum_{t=1}^T (D f_t^k)^\top D f_t^k$ and observe that $\hat \Sigma_K^\pi = \sum_{k=1}^K \hat \Sigma^{\pi,k}$. 
    Observe that for any unit vector $v$, $\E\brac{v^\top (\hat \Sigma_K^{\pi,k} - \Sigma^\pi) v} = 0$. Furthermore, using the bound on $\norm{D f_t^n}_{\op}$ in Assumption~\ref{asmp: smooth dynamics}, we have $0 \leq v^\top \hat \Sigma^{\pi,k} v \leq T L_f^2$. Then by Hoeffding's inequality (Theorem 2.2.6 of \cite{vershynin2018high}), we have
    \begin{align*}
        \P\brac{\abs{v^\top (\hat \Sigma_K^\pi - \Sigma^\pi) v \geq \frac{\rho}{2}}} \leq \exp\paren{\frac{-\rho^2}{2 K T^2 L_f^4}}.
    \end{align*}
    Inverting the bound, we find that with probability at least $1-\delta$, 
    \begin{align*}
        \norm{\hat \Sigma_K^\pi - \Sigma^\pi} \leq \frac{\sqrt{2} T L_f^2 \sqrt{3 d_{\phi} + \log(1/\delta)}}{\sqrt{K}}. 
    \end{align*}
    The result follows from the lower bound on $K$. 
\end{proof}

Using the above concentration result and self-normalized martingale bound, we may prove our main bound on the system identification error, \Cref{thm: identification error bound}.
\begin{proof}
We will invoke the covariance concentration lemma above with the value of $\beta$ chosen according to \Cref{lem: sn martingale bnd}. In particular, we define 
\begin{align*}
    \tau_{\mathsf{cov\,conc}}(\delta) =\frac{2 T^2 L_f^4 (3 d_{\phi} + \log(1/\delta))}{\beta^2}. 
\end{align*}
If $K\geq\tau_{\mathsf{cov\,conc}}(\delta/4)$, we may invoke \Cref{lem: covariance concentration} to show that with probability at least $1-\delta/4$, 
\begin{align*}
    \norm{\hat \Sigma_K^\pi - \Sigma^\pi} \leq \beta.
\end{align*} By the fact that $\beta\leq\frac{\lambda_{\min}(\Sigma^\pi)}{4}$, we have that $\lambda_{\min}(\hat \Sigma_K^\pi) \geq \frac{1}{2} \lambda_{\min}(\Sigma^\pi)$. Given this condition and the burn-in conditions $K \geq \max\curly{\tau_{\ER}(\delta/8), \tau_{\mathsf{small\,error}}(\delta/4), \tau_{\mathsf{h.o.t.}}(\delta/4, \beta)}$, we satisfy the conditions of \Cref{lem:deltamethod}. Therefore, conditioned on the event of covariance concentration, we have that with probability at least $1-\delta/4$, 
\begin{align*}
        \hat \phi_K - \phi^\star &= (I+G)(\hat \Sigma_K^\pi)^{-1} \frac{1}{K} \sum_{t,k=1}^{T,K} (D f_t^k)^\top W_t^k,
    \end{align*}
    where $G$ is a matrix satisfying $\norm{G}_{\op} \leq \beta $.  Then by the triangle inequality, we have that under the success events of covariance concentration and the delta method, 
    \begin{align}
        \label{eq: param recovery by snm}
        \norm{\hat \phi_K - \phi^\star}_H^2 \leq \paren{\norm{(\hat \Sigma_K^\pi)^{-1} \frac{1}{K} \sum_{t,k=1}^{T,K} (D f_t^k)^\top W_t^k}_H + \beta \norm{H}^{1/2} \norm{(\hat \Sigma_K^\pi)^{-1} \frac{1}{K} \sum_{t,k=1}^{T,K} (D f_t^k)^\top W_t^k}}^2.
    \end{align}
    We have reduced the parameter recovery errort to the sum of weighted norms of a self-normalized martingale. Calling on self-normalized martingale machinary \cite{abbasi2011improved}. In particular, invoking \Cref{lem: sn martingale bnd} with failure probability $\delta/4$ and union bounding provides that with probability at least $1-\delta/2$, the following bounds hold: 
    \begin{align*}
        \norm{(\hat \Sigma_K^\pi)^{-1} \frac{1}{K} \sum_{t,k=1}^{T,K} (D f_t^k)^\top W_t^k}_H \leq \sqrt{2 \sigma_w^2 (1+\xi) \trace( \paren{K \Sigma^\pi}^{-1} H) + 4 \sigma_w^2 (1+\xi) \norm{(K \Sigma^\pi)^{-1} H} \log\frac{4}{\delta}},\\
        \norm{(\hat \Sigma_K^\pi)^{-1} \frac{1}{K} \sum_{t,k=1}^{T,K} (D f_t^k)^\top W_t^k} \leq \sqrt{2 \sigma_w^2 (1+\xi) \trace( \paren{K \Sigma^\pi}^{-1}) + 4 \sigma_w^2 (1+\xi) \norm{(K \Sigma^\pi)^{-1}} \log\frac{4}{\delta}},
    \end{align*}
    where $\xi  = \frac{4 \beta}{\lambda_{\min}(\Sigma^\pi)}$ 
    Substituting these bounds into \eqref{eq: param recovery by snm}, we find that by union bounding over the success events, we have that with probability at least $1-\delta$, 
    \begin{align*}
        \norm{\hat \phi_K - \phi^\star}_H^2 \leq 2 \sigma_w^2 (1+\xi + 4 \beta d_{\phi}) \trace( \paren{K \Sigma^\pi}^{-1} H) + 4 \sigma_w^2 (1+\xi + 4 \beta d_{\phi}) \norm{(K \Sigma^\pi)^{-1} H} \log\frac{4}{\delta}. 
    \end{align*}
\end{proof}

\section{Results for smooth nonlinear systems}
\label{s: smooth systems}

We begin by modifying the result of Lemma D.1 of  \cite{wagenmaker2023optimal} that bounds the derivatives of the cost with respect to both the controller and the dynamics parameters.

\begin{lemma}
    \label{lem: differentiable objective}
    Under Assumptions~\ref{asmp: smooth dynamics}, \ref{asmp: bounded costs}, and \ref{asmp: smooth policy class} we have that for any $\phi \in \calB(\phi^\star, r_{\cost}(\phi^\star))$, the controller cost $\calJ(\pi^{\theta}, \phi)$ is four times differentiable in $\theta$ and $\phi$, and the derivatives satisfy the bound
    \begin{align*}
        \norm{D_\phi^{(i)} D_{\theta}^{(j)} \calJ(\pi^\theta, \phi)}_{\op} \leq \mathsf{poly}(L_f, L_{\theta}, L_{\cost}, \sigma_w^{-1}, T, \dx). 
    \end{align*}
    for all $i,j \in \curly{0,1,2,3}$, $1 \leq i+j \leq 4$. 
\end{lemma}
\begin{proof}
    The proof proceeds as in that of Theorem D.1 of \cite{wagenmaker2023optimal}. 

    In particular, we first define $p_w$ to denote the density of the noise, $W_t$. Let $p_{\phi, \theta}(\cdot)$ denote the density over the states $X_{2:T+1}$ induced by playing controller $\pi^\theta$ on the system \eqref{eq: dyn} with parameter $\phi$. Observe that 
    \begin{align}
        \label{eq: traj density}
        p_{\theta, \phi}(X_{2:T+1}) = \prod_{t=1}^T p_w(X_{t+1} - f(X_t, \pi^{\theta}_t(X_t); \phi^\star)).
    \end{align}
    We may express our average controller cost as 
    \begin{align*}
        J(\pi^{\theta}, \phi) = \int \paren{\sum_{t=1}^T c_t(X_t, \pi_{\theta}(X_t)) + c_{T+1}(X_{T+1})} p_{\theta, \phi}(X_{1:T+1}) d X_{2:T+1}, 
    \end{align*}
    with $X_1$ fixed as $0$. 

    Note that by our noise assumptions and the smoothness assumptions, we may interchange limits and integrals by appealing to the dominated convergence theorem.  

    Let $\phi_{\boldsymbol m} = \phi + m_1 \Delta_1^\phi + m_2 \Delta_2^\phi + m_3 \Delta_3^\phi$ and $\theta_{\boldsymbol{s}} = \theta + s_1 \Delta_1^\theta + s_2 \Delta_2^\theta  + s_3 \Delta_3^\theta$ for $\Delta_i^\theta$ and $\Delta_j^\phi$ satisfying $\norm{\Delta_i^\theta} = 1$ and $\norm{\Delta_j^\phi} =1$ for $i,j=1, 2, 3$. We will consider the directional derivatives of $ \calJ(\pi^{\theta}, \phi)$ by considering the derivatives of $\calJ(\pi^{\theta_{\boldsymbol s}}, \phi_{\boldsymbol{m}})$ with respect to $x_1, x_2, x_3, x_4 \in \curly{m_1, m_2, m_3, s_1, s_2, s_3}$.  By showing that $\calJ(\pi^{\theta_{\boldsymbol s}}, \phi_{\boldsymbol{m}})$ is differentiable with respect to these quantities for any unit vectors $\Delta_i^\theta$ and $\Delta_i^{\phi}$, then $\calJ(\pi^\theta, \phi)$ is four times differentiable with respect to $\theta$ and $\phi$. Additionally, we can bound the operator norm of the derivatives by bounding the directional derivatives for all unit vectors   $\Delta_i^\theta$ and $\Delta_i^{\phi}$.  

    To show differentiability, note that
    \begin{align*}
        &\frac{d}{dx_1} J(\pi^{\theta_{\bolds}}, \phi_{\boldm}) = \int \frac{d}{dx_1} \paren{\sum_{t=1}^T  c_t(X_t, \pi_{\theta_{\bs}}(X_t)) + c_{T+1}(X_{T+1})} p_{\theta_{\bolds}, \phi_{\boldm}}(X_{1:T+1}) dX_{2:T+1} .  \\
        & = \int \underbrace{\paren{\sum_{t=1}^T  \nabla_u c_t(X_t, \pi_{\theta_{\bs}}(X_t))^\top \frac{d}{d x_1} \pi_{\theta_{\bs}}(X_t)}}_{\mbox{Cost gradient}} p_{\theta_{\bolds}, \phi_{\boldm}}(X_{1:T+1}) d X_{2:T+1} \\
        &+\int \underbrace{\paren{\frac{d}{dx_1} \log\paren{p_{\theta_{\bolds}, \phi_{\boldm}} (X_{1:T+1} )}}}_{\mbox{Derivative of log density}} \paren{\sum_{t=1}^T  c_t(X_t, \pi_{\theta}(X_t)) + c_{T+1}(X_{T+1})} p_{\theta_{\bolds}, \phi_{\boldm}}(X_{1:T+1}) d X_{2:T+1}.
    \end{align*}
    Proceeding to differentiate the above quantity with respect to $x_2$, $x_3$, and $x_4$, provides from the product rule a sum of terms which are of the form above, but with the derivative of the log density and the cost gradient replaced by the product of higher order derivatives of the log density.  

    Recalling the form of the density in \eqref{eq: traj density}, we have that 
    \begin{align*}
        \log\paren{p_{\theta_{\bolds}, \phi_{\boldm}} (X_{2:T+1})} = \sum_{t=1}^T -\frac{1}{2\sigma_w^2} \norm{X_{t+1} - f(X_t, \pi^{\theta_{\bolds}}; \phi_{\boldm})}^2 + 
        T\log\frac{1}{\sqrt{2\pi \sigma_w^{2\dx}}}.  
    \end{align*}
    By the differentiability of $f$ with respect to $\phi$ and $u$ from Assumption~\ref{asmp: smooth dynamics} and the differentiability of $\pi^{\theta}$ with respect to $\theta$ from Assumption~\ref{asmp: smooth policy class}, we have that the log density above is four times differentiable in $x_1, \dots x_4$ for any values of $\Delta_i^\theta$ and $\Delta_j^{\phi}$. Similarly, by differentiability of the stage costs from Assumption~\ref{asmp: bounded costs} and the policy from Assumption~\ref{asmp: smooth policy class}, the cost gradient is three times differentiable in $x_1,\dots, x_4$. 

    To bound the norm of the gradients, observe that
    \begin{align*}
        \frac{d}{d m_i} \log\paren{p_{\theta_{\bolds}, \phi_{\boldm}} (X_{2:T+1})}\vert_{\bolds=\boldm=0} &= \sum_{t=1}^T \frac{1}{\sigma_w^2} \paren{X_{t+1} - f(X_t, \pi^{\theta}; \phi)}^\top D_{\phi} f(X_t, \pi^{\theta; \phi}) \Delta_i^{\phi} \\
        \frac{d}{d s_i} \log\paren{p_{\theta_{\bolds}, \phi_{\boldm}} (X_{2:T+1})}\vert_{\bolds=\boldm=0} &= \sum_{t=1}^T \frac{1}{\sigma_w^2} \paren{X_{t+1} - f(X_t, \pi^{\theta}; \phi)}^\top D_{u} f(X_t, \pi^{\theta}; \phi) D_{\theta} \pi^{\theta}  \Delta_i^{\theta}
    \end{align*}
    Differentiating further with respect to $m_i$ and $s_i$ results only in higher order derivatives. Using the operator norm bounds on the derivative of the policy and the dynamics, we may conclude that for $i,j \in \curly{0,1,2,3}$ such that $1\leq i+j\leq 4$,
    \begin{align*}
        \norm{D_{\phi}^{i} D_{\theta}^{j} \log\paren{p_{\theta, \phi} (X_{2:T+1})}}_{\op} \leq \mathsf{poly}\paren{L_f, L_{\theta}, \sigma_w^{-1}} \sum_{t=1}^T \paren{1 + \frac{\norm{X_{t+1} - f(X_t, \pi^{\theta}(X_t); \phi)}}{\sigma_w}}
    \end{align*}

    We may then in turn bound $\norm{D_\phi^{(i)} D_{\theta}^{(j)} \calJ(\pi^\theta, \phi)}_{\op}$ as 
    \begin{align*}
        &\norm{D_\phi^{(i)} D_{\theta}^{(j)} \calJ(\pi^\theta, \phi)}_{\op} \leq \mathsf{poly}\paren{L_f, L_{\theta}, \sigma_w^{-1}} \\
        &\times\int \paren{\sum_{t=1}^T \paren{1 + \frac{\norm{X_{t+1} - f(X_t, \pi^{\theta}(X_t); \phi)}}{\sigma_w}}}^4 \paren{\sum_{t=1}^T  c_t(X_t, \pi_{\theta}(X_t) + c_{T+1}(X_{T+1})} p_{\theta, \phi}(X_{1:T+1}) d X_{2:T+1} \\
        & \leq \mathsf{poly}\paren{L_f, L_{\theta}, \sigma_w^{-1}} \sqrt{\int \paren{\sum_{t=1}^T  c_t(X_t, \pi_{\theta}(X_t)) + c_{T+1}(X_{T+1})}^2 p_{\theta, \phi}(X_{1:T+1}) d X_{2:T+1}} \\
        &\times \sqrt{\int \paren{\sum_{t=1}^T \paren{1 + \frac{\norm{X_{t+1} - f(X_t, \pi^{\theta}(X_t); \phi)}}{\sigma_w}}}^8 p_{\theta, \phi}(X_{1:T+1}) d X_{2:T+1}},
    \end{align*}
    by Cauchy-Schwarz. We have from 
    Assumption~\ref{asmp: bounded costs} that \begin{align*}
        &\sqrt{\int \paren{\sum_{t=1}^T  c_t(X_t, \pi_{\theta}(X_t)) + c_{T+1}(X_{T+1})}^2 p_{\theta, \phi}(X_{1:T+1}) d X_{2:T+1}}=\sqrt{\E\brac{\paren{\sum_{t=1}^T c_t(X_t, \pi_{\theta}(X_t)) + C_{T+1}(X_T)}^2}} \leq L_{\cost}. 
    \end{align*}
    From Lemma $A.1$ of \citet{wagenmaker2023optimal} ,
    \begin{align*}
        \sqrt{\int \paren{\sum_{t=1}^T\paren{1 + \frac{\norm{X_{t+1} - f(X_t, \pi^{\theta}(X_t); \phi)}}{\sigma_w}}}^8 p_{\theta, \phi}(X_{1:T+1}) d X_{2:T+1}} \leq \mathsf{poly}(T, \dx).
    \end{align*}
    We may similarly bound the cost gradient by a polynomial in $L_{\cost}$ and $L_{\theta}$.
    Therefore our derivative is bounded as 
    \begin{align*}
       \norm{D_\phi^{(i)} D_{\theta}^{(j)} \calJ(\pi^\theta, \phi)}_{\op} \leq \mathsf{poly}\paren{T, L_f, L_{\theta}, \sigma_w^{-1}, L_{\cost}, \dx}.
    \end{align*}
\end{proof}

Given \Cref{lem: differentiable objective},  Assumption~\ref{asmp:  smooth CE} may be shown to hold if $\nabla_{\theta}^2 J(\pi^\theta, \phi^\star) \succ 0$ using an argument identical to that of Proposition 6 in \cite{wagenmaker2023optimal}, with \Cref{lem: differentiable objective} invoked in place of their Lemma D.1. We may also use \Cref{lem: differentiable objective} to show the following two results from \citet{wagenmaker2023optimal}.

\begin{lemma}[Lemma D.3 of \citet{wagenmaker2023optimal}]
    \label{lem: model task hessian error}
    Under Assumptions~\ref{asmp: smooth dynamics}, \ref{asmp: smooth policy class}, \ref{asmp: bounded costs} and \ref{asmp: smooth CE} we have that for any $\phi \in \calB(\phi^\star, \min\curly{r_{\cost}(\phi^\star), r_{\theta}(\phi^\star)})$,
    \begin{align*}
        \norm{\calH(\phi^\star) - \calH(\phi)}_{\op} \leq C_{\mathsf{Hpert}} \norm{\phi - \phi^\star}. 
    \end{align*}
    where
    \begin{align*}
        C_{\mathsf{Hpert}} = \mathsf{poly}\paren{L_{\pi^\star}, L_f, L_{\theta}, L_{\cost}, \sigma_w^{-1},  T, \dx}. 
    \end{align*}
\end{lemma}

\begin{lemma}[Lemma D.4 of \citet{wagenmaker2023optimal}]
    \label{lem: model task hessian bound}
    Under Assumptions~\ref{asmp: smooth dynamics}, \ref{asmp: smooth policy class},\ref{asmp: bounded costs} and \ref{asmp: smooth CE} we have that for any $\phi \in \calB(\phi^\star, \min\curly{r_{\cost}(\phi^\star), r_{\theta}(\phi^\star)})$,
    \begin{align*}
        \norm{\calH(\phi)}_{\op} \leq \mathsf{poly}\paren{L_{\pi^\star}, L_f, L_{\theta}, L_{\cost}, \sigma_w^{-1},  T, \dx}. 
    \end{align*}
\end{lemma}

The proofs of \Cref{lem: cost decomposition}, \Cref{lem: model task hessian error}, and \Cref{lem: model task hessian bound} from the same arguments as the proofs of Lemma D.2, Lemma D.3 and Lemma D.4 of \cite{wagenmaker2023optimal} with \Cref{lem: differentiable objective} taking place of their Lemma D.1. 

Finally, we present a result that demonstrates that under our smoothness assumptions, the Fisher information is differentiable with respect to the dynamics parameter, and that the 
derivative is bounded by a quantity which is polynomial in relevant system parameters and the length of the horizon.

\begin{lemma}
    \label{lem: Fisher derivative}
    For any exploration policy $\pi \in \Pi_{\exp}$, any $\phi \in \R^{d_{\phi}}$, and any unit vector $v\in \R^{d_{\phi}}$ we have that 
     $v^\top \FI^\pi(\phi) v$ is differentiable with respect to $\phi$, and that 
     \begin{align*}
         \norm{\nabla_{\phi} v^\top \FI^\pi(\phi) v} \leq 3\frac{T^2 \max\curly{L_f,1}^3 \sqrt{\dx}}{\min\curly{1,\sigma_w}}. 
     \end{align*}
\end{lemma}
\begin{proof}
    Express
    \begin{align*}
        \nabla_{\phi} v^\top \FI^\pi(\phi) v = \nabla_{\phi} \int \sum_{t=1}^T v^\top D f(X_t, \pi(X_t); \phi) D f(X_t, \pi(X_t); \phi)^\top v p_{\phi}(X_{1:T+1}) d_{X_{2:T+1}},
    \end{align*}
    where $X_1=0$ and $
        p_{\phi}(X_{1:T+1}) = \Pi_{t=1}^T p_w(X_{t+1} - f(X_t, \pi(X_{1:t}); \phi)).$

    The gradient is then given by 
    \begin{align*}
        \nabla_{\phi} v^\top \FI^\pi(\phi) v &=  \int \paren{\nabla_{\phi}\sum_{t=1}^T v^\top D f(X_t, \pi(X_{1:t}); \phi) D f(X_t, \pi(X_{1:t}); \phi)^\top v}  p_{\phi}(X_{1:T+1}) d_{X_{2:T+1}} \\
        & + \int \sum_{t=1}^T v^\top D f(X_t, \pi(X_{1:t}); \phi) D f(X_t, \pi(X_{1:t}); \phi)^\top v \paren{\nabla_{\phi} \log\paren{p_{\phi}(X_{1:T+1})}} p_{\phi}(X_{1:T+1}) d_{X_{2:T+1}}.
    \end{align*}
    By Assumption~\ref{asmp: smooth dynamics}, the first term is differentiable, and the norm is bounded by $2T L_f^2$. For the second term, observe that 
    \begin{align*}
        \nabla_{\phi} \log\paren{p_{\phi}(X_{1:T+1})} = \nabla_{\phi}\paren{-\frac{1}{2\sigma_w^2} \sum_{t=1}^T \norm{X_{t+1} - f(X_t, \pi(X_{1:T}); \phi)}^2} = \frac{1}{\sigma_w^2} \sum_{t=1}^T D_{\phi} f(X_t, \pi(X_{1:t}; \phi) W_t.
    \end{align*}
    Therefore, the norm of the second term may be bounded by $\frac{T^2 L_f^3 \sqrt{\dx}}{\sigma_w}$ by appealing to the triangle inequality, submultiplicativity, the bound $\mathbf{E}[\norm{g_t}] \leq \sqrt{n}$ if $g_t$ is a standard multi-variate random vector of dimension $n$. 
    
\end{proof}
\section{Proof of Main Excess Cost Bound, \Cref{thm: main}}

First we present a result stating that using the estimated parameter in place of the true parameter leads to an adequate approximation for the ideal exploration objective.

\begin{lemma}
    \label{lem: cov and H pert}
    Consider a policy $\pi \in \Pi_{\exp}$. Suppose $\norm{\hat \phi - \phi^\star} \leq \frac{\tilde \lambda \min\curly{\sigma_w,1}}{6 T^2 \max\curly{1,L_f}^3 \sqrt{\dx} }$ for some $\tilde \lambda > 0$. Assume that $\FI^{\pi}(\hat\phi) \succeq \tilde \lambda I$. 
    Then $\FI^\pi(\phi^\star) \succeq \frac{\tilde\lambda}{2} I$, and
    \begin{enumerate}
        \item $\trace\paren{\calH(\phi^\star) \FI^\pi(\phi^\star)^{-1}} \leq \trace\paren{\calH(\hat\phi) \FI^\pi(\hat\phi)^{-1}} + \mathsf{poly}\paren{L_{\pi^\star}, L_f, L_{\theta}, L_{\cost}, \sigma_w^{-1},   T, \dx, d_{\phi}, \norm{\calH(\phi^\star)}, \frac{1}{\tilde\lambda}} \norm{\hat \phi - \phi^\star}$
        \item $\trace\paren{\calH(\hat\phi) \FI^\pi(\hat\phi)^{-1}} \leq\trace\paren{\calH(\phi^\star) \FI^\pi(\phi^\star)^{-1}} + \mathsf{poly}\paren{L_{\pi^\star}, L_f, L_{\theta}, L_{\cost}, \sigma_w^{-1},  T, \dx, d_{\phi}, \norm{\calH(\phi^\star)}, \frac{1}{\tilde\lambda}} \norm{\hat \phi - \phi^\star}$
    \end{enumerate}
\end{lemma}
\begin{proof}\,\\
Begin by bounding $\norm{\FI^\pi(\hat\phi) - \FI^\pi(\phi^\star)} = \sup_{v\in\R^{\dtheta}, \norm{v}=1} v^\top \FI^\pi(\hat\phi) v - v^\top \FI^\pi(\phi^\star) v$. By a first order Taylor expansion, it holds that for any unit vector $v\in\R^{\dtheta}$, 
\begin{align*}
   v^\top \FI^\pi(\hat\phi) v - v^\top \FI^\pi(\phi^\star) v \leq \norm{\nabla_{\phi} v^\top \FI^\pi(\tilde \phi) v} \norm{\hat \phi - \phi^\star}.
\end{align*}
for some $\tilde\phi = \alpha \hat \phi + (1-\alpha) \phi^\star$, $\alpha \in [0,1]$.
By \Cref{lem: Fisher derivative}, it holds that 
\[
    \norm{\nabla_{\phi} v^\top \FI^\pi(\tilde \phi) v} \leq 3\frac{T^2 \max\curly{L_f,1}^3 \sqrt{\dx}}{\min\curly{1,\sigma_w}}. 
\]
Then by the assumed bound on $\norm{\hat \phi-\phi^\star}$, we find that $\norm{\FI^\pi(\hat\phi) - \FI^\pi(\phi^\star)} \leq \frac{\tilde\lambda}{2}$. Therefore, $\FI^\pi(\phi_\star) \succeq \frac{\tilde \lambda}{2}$. 

We also have that 
\begin{align*}\norm{\FI^\pi(\hat\phi)^{-1} - \FI^\pi(\phi^\star)^{-1}} \leq \frac{1}{\lambda_{\min}\paren{\FI^\pi(\hat\phi)} \lambda_{\min}\paren{\FI^\pi(\phi^\star)}} \norm{\FI^\pi(\hat\phi) - \FI^\pi(\phi^\star)} \leq \frac{6 T^2 \max\curly{L_f,1}^3 \sqrt{\dx}}{\tilde\lambda^2 \min\curly{1,\sigma_w}} \norm{\hat\phi-\phi^\star}
\end{align*}
Furthermore, by \Cref{lem: model task hessian error},  $\norm{\calH(\hat\phi) - \calH(\phi^\star)}\leq \mathsf{poly}\paren{L_{\pi^\star}, L_f, L_{\theta}, L_{\cost}, \sigma_w^{-1}, T, \dx} \norm{\phi - \phi^\star}. $ 

To conclude, we bound 
\begin{align*}
    \abs{\trace\paren{\calH(\phi^\star) \FI^\pi(\phi^\star)^{-1}} - \trace\paren{\calH(\phi^\star) \FI^\pi(\hat\phi)^{-1}}} &\leq \norm{\calH(\phi^\star)} d_{\phi} \norm{\FI^\pi(\phi^\star)^{-1} - \FI^\pi(\hat\phi)^{-1}} \\&\leq \norm{\calH(\phi^\star)} d_{\phi}\frac{6 T^2 \max\curly{L_f,1}^3 \sqrt{\dx}}{\tilde\lambda^2 \min\curly{1,\sigma_w}} \norm{\hat\phi-\phi^\star}
\end{align*}
and
\begin{align*}
    \abs{\trace\paren{\calH(\phi^\star) \FI^\pi(\hat\phi)^{-1}}  - \trace\paren{\calH(\phi^\star) \FI^\pi(\hat\phi)^{-1}}}  &\leq \frac{d_{\phi}}{\tilde\lambda} \norm{\calH(\hat\phi)-\calH(\phi^\star)} \\&\leq \frac{d_{\phi}}{\tilde\lambda}  \mathsf{poly}\paren{L_{\pi^\star}, L_f, L_{\theta}, L_{\cost}, \sigma_w^{-1}, T, \dx} \norm{\phi - \phi^\star}.
\end{align*}
Combining these inequalities with the triangle inequality provides the desired result. 
\end{proof}

Armed with the above results, we are ready to prove \Cref{thm: main}.

\mainresult*
\begin{proof} 

    We have that $\gamma N\geq \tau_{\ER}(\delta/2)$. Then by Assumption~\ref{asmp: loja}, we may invoke \Cref{cor: slow rate} to show that with probability at least $1-\delta/4$, $\norm{\hat\phi^--\phi^\star}\leq \frac{\paren{\frac{512 \sigma_w^2}{T } \paren{\dx +d_{\phi}\log\paren{L_f T N} + \log\frac{4}{\delta}}}^{\alpha}}{(\gamma N) ^\alpha} \triangleq \frac{R}{(\gamma N)^{\alpha}}$. Denote the event that this inequality holds as $\calE_{\coarse}$. Condition on the event that this event holds for the remainder of the proof. 
    
    Let $C_{\mathsf{H pert}}$ be the system theoretic constant defined in \Cref{lem: model task hessian error}. As long as $
        \gamma N \geq \frac{ 2 R C_{\mathsf{H pert}}}{\lambda_{\min}(\calH(\phi^\star))
        }^{\frac{1}{\alpha}}$, it follows that
    \begin{align}
        \label{eq: coarse bound after burn-in}
        \norm{\hat\phi^- - \phi^\star} \leq \frac{\lambda_{\min}(\calH(\phi^\star))}{2 C_{\mathsf{H pert}}}. 
    \end{align}
    Then by \Cref{lem: model task hessian error}, we have $
        \lambda_{\min}(\calH(\hat\phi^-)) \geq \frac{1}{2} \lambda_{\min}(\calH(\phi^\star))$ and $\norm{\calH(\hat\phi^-)} \leq 2\norm{\calH(\phi^\star)}$. Additionally, suppose $\gamma N \geq \paren{\frac{6 R  T^2 \max\curly{L_f, 1}^3 \sqrt{\dx}}{\min\curly{\sigma_w, 1} \mu }}^{1/\alpha}$. Then by a first order Taylor expansion combined with \Cref{lem: Fisher derivative}, we have that for any policy $\pi\in\Pi_{\exp}$ $\norm{\FI^\pi(\hat\phi^-) -\FI^\pi(\phi^\star)} \leq \frac{\mu}{2}$. As a result, by Assumption~\ref{asmp: good policy}, there exists a policy $\pi$ such that $\lambda_{\min}(\FI^\pi(\hat\phi^-)) \geq \frac{\mu}{2}$. Therefore, if $\pi_{\exp}$ satisfies $\trace((\calH(\hat\phi^-) + \nu I) \FI^\pi(\hat\phi^-)^{-1}) \leq (1+\varepsilon) \trace((\calH(\hat\phi^-) + \nu I) \FI^{\tilde \pi}(\hat\phi^-)^{-1}) $ for all $\tilde \pi \in \Pi_{\exp}$, it holds that
        \begin{align*}
            \frac{\lambda_{\min}(\calH(\hat\phi^\star)) + \nu}{2 \lambda_{\min}(\FI^{\pi_{\exp}}(\hat\phi^-))} \leq \trace\paren{\paren{\calH(\hat\phi^-) + \nu} \FI^{\pi_{\exp}}(\hat\phi^-)^{-1}} \leq \frac{(1+\varepsilon) 4 d_{\phi} \paren{\norm{\calH(\phi^\star)} + \nu}} {\mu},
        \end{align*}
        and therefore that 
        \begin{align*}
            \lambda_{\min}(\FI^{\pi_{\exp}}(\hat\phi^-)) \geq \frac{\mu \paren{\lambda_{\min}(\calH(\phi^\star)) + \nu}}{8 (1+\varepsilon)d_{\phi} \paren{\norm{\calH(\phi^\star)} + \nu}}.
        \end{align*}

    Suppose
    \[\gamma N \geq \paren{\frac{(1+\varepsilon)48 R  T^2 \max\curly{L_f, 1}^3 \sqrt{\dx} d_{\phi} \paren{\norm{\calH(\phi^\star)} +\nu}}{\min\curly{\sigma_w, 1} \paren{\lambda_{\min}(\calH(\phi^\star)) + \nu}\mu }}^{1/\alpha}. \]

    Then by \Cref{lem: cov and H pert}, it follows that $\FI^{\pi_{\mix}}(\phi^\star) \succeq \frac{1}{2} \FI^{\pi_{\exp}}(\phi^\star) \succeq  \frac{\mu \paren{\lambda_{\min}(\calH(\phi^\star)) + \nu}}{32 (1+\varepsilon) d_{\phi} \paren{\norm{\calH(\phi^\star)} + \nu}} I$. 
     Additionally, the policy $\pi_\mix$ is $( \frac{1}{\gamma} C_{\loja},\alpha)$-Lojasiewicz. This can be seen by noting that by the law of total expectation. In particular, for all $\phi \in \R^{d_{\phi}}$,
     \begin{align*}
         \ER_{\pi_\mix, \phi^\star}(\phi) \geq \frac{1}{\gamma} \ER_{\pi^0, \phi^\star}(\phi).
     \end{align*}

    Then the slow rate from \Cref{cor: slow rate} ensures that playing the policy $\pi^{\mix}$ for $(1-\gamma) N$ episodes results in a refined estimate that satisfies $\norm{\hat \phi^+ - \phi^\star} \leq \frac{R}{((1-\gamma) N)^\alpha}$. Additionally, under the  event $\calE_{\coarse}$ in conjunction with the burn-in requirement
    \begin{align*}
        (1-\gamma) N &\geq \mathsf{poly}_\alpha\bigg(T,L_f,  d_{\phi}, \dx,  \sigma_w, \log N,  \log\frac{1}{\delta}, \log B, C_{\loja}, \frac{1}{\beta} \bigg),
    \end{align*}
    the conditions are set to apply \Cref{thm: identification error bound} on the data collected by running $\pi^{\mix}$ for $(1-\gamma) N$ episodes. To do so, we note that for any $\beta \in \paren{0, \frac{\mu \paren{\norm{\calH(\phi^\star)} + \nu}}{512 d_{\phi} (\lambda_{\min}(\calH(\phi^\star)) + \nu)}}$,
    the condition $\beta \leq \frac{\sigma_w^2 \lambda_{\min}(\FI^{\pi_{\mix}}(\phi^\star)}{4}$
    is met by the lower bound on $\FI^{\pi_{\mix}}(\phi^\star)$. 
    Therefore, we may invoke \Cref{thm: identification error bound} to find that with probability at least $1-\delta/2$, 
    \begin{equation}
    \begin{aligned}
        \label{eq: fine identification bound}
         \norm{\hat \phi^+ - \phi^\star}_{\calH(\phi^\star) + \nu I}^2 &\leq   2 \frac{(1+\xi)}{1-\gamma} \paren{ \frac{\trace\paren{ \paren{\calH(\phi^\star) + \nu I} \FI^\pi(\phi^\star)^{-1}} }{N} + 2 \frac{\norm{\paren{\calH(\phi^\star) + \nu I}  \FI^\pi(\phi^\star)^{-1}} }{N} \log\frac{8}{\delta}}  \\
         &\leq \frac{(1+\xi)}{1-\gamma} \paren{2 + 4 \log\frac{8}{\delta}} \frac{\trace\paren{ \paren{\calH(\phi^\star) + \nu I} \FI^\pi(\phi^\star)^{-1}}}{N},
    \end{aligned}
    \end{equation}
    where $\xi = 4\beta \paren{\frac{128 d_{\phi} \paren{\norm{\calH(\phi^\star)} +\nu}}{\mu \paren{\lambda_{\min}(\calH(\phi^\star)) + \nu}} + d_{\phi}}$.
    Denote this event as $\calE_{\mathsf{ID}}$. It follows from \Cref{lem: cov and H pert} that
    \begin{equation}
    \label{eq: opt obj under approx under mix}
    \begin{aligned}
        \trace(&\paren{\calH(\phi^\star) + \nu I }\FI^{\pi_\mix}(\phi^\star)^{-1}) \leq \trace\paren{\paren{\calH(\hat\phi^-) + \nu I } \FI^{\pi_{\mix}}(\hat\phi^-)^{-1}} \\
        &+ \mathsf{poly}\paren{L_{\pi^\star}, L_f, L_{\theta}, L_{\cost}, \sigma_w^{-1}, T, \dx, d_{\phi}, \norm{\calH(\phi^\star)}, \nu, \frac{1}{\mu}, \frac{1}{\lambda_{\min}(\calH(\phi^\star)) + \nu}} \norm{\hat\phi^--\phi^\star}. 
    \end{aligned}
    \end{equation}
    By the fact that $\FI^{\pi_{\mix}}(\hat\phi^-) \succeq (1-\gamma) \FI^{\pi_{\exp}}(\hat\phi^-)$, and the fact that $\pi_{\exp}$ optimizes the approximate exploration objective, it holds that $\trace\paren{\paren{\calH(\hat\phi^-) + \nu I } \FI^{\pi_{\mix}}(\hat\phi^-)^{-1}} \leq \frac{1 +\varepsilon}{1-\gamma} \inf_{\pi\in\Pi_{\exp}} \trace\paren{\paren{\calH(\hat\phi^-) + \nu I }\FI^\pi(\hat\phi^-)^{-1}}$. Then by again appealing to \Cref{lem: cov and H pert}, it holds that 
    \begin{equation}
    \label{eq: estimated objective by optimal objective}
    \begin{aligned}
         \inf_{\pi\in\Pi_{\exp}} &\trace\paren{\paren{\calH(\hat\phi^-) + \nu I 
         } \FI^\pi(\hat\phi^-)^{-1}} \leq \inf_{\pi\in\Pi_{\exp}} \trace\paren{\paren{\calH(\phi^\star) + \nu I } \FI^\pi(\phi^\star)^{-1}} \\
         &+   \mathsf{poly}\paren{L_{\pi^\star}, L_f, L_{\theta}, L_{\cost}, \sigma_w^{-1}, T, \dx, d_{\phi}, \norm{\calH(\phi^\star)},\nu, \frac{1}{\mu}, \frac{1}{\lambda_{\min}(\calH(\phi^\star)) + \nu}} \norm{\hat\phi^--\phi^\star}.
    \end{aligned}
    \end{equation}
  
    Combining \eqref{eq: opt obj under approx under mix} with \eqref{eq: estimated objective by optimal objective} and using the bound $\norm{\hat \phi^- - \phi^\star} \leq \frac{R}{(\gamma N)^\alpha}$ yields 
    \begin{equation}
    \label{eq: mixed policy by optimal}
    \begin{aligned}
         &\trace(\paren{\calH(\phi^\star) + \nu} \FI^{\pi_\mix}(\phi^\star)^{-1}) \leq \frac{1 + \varepsilon}{1-\gamma} \inf_{\pi\in\Pi_{\exp}} \trace\paren{\paren{\calH(\phi^\star) + \nu I } \FI^\pi(\phi^\star)^{-1}} \\
         &+ \mathsf{poly}\paren{L_{\pi^\star}, L_f, L_{\theta}, L_{\cost}, \sigma_w^{-1}, \sigma_w, T, \dx, d_{\phi}, \norm{\calH(\phi^\star)}, \nu, \frac{1}{\mu}, \frac{1}{\lambda_{\min}(\calH(\phi^\star)) + \nu}}  \frac{R}{(\gamma N)^\alpha}.
    \end{aligned}
    \end{equation}
    As a result of the above identfiication bounds, the excess cost of the returned policy $\pi^\star(\hat\phi^+)$ is bounded via \Cref{lem: cost decomposition} as
    \begin{align*}
        \calJ(\phi^+) - \calJ(\phi^\star) &\leq  \norm{\hat \phi^+ - \phi^\star}^2_{\calH(\phi^\star)} +  C_{\cost} \norm{\hat \phi^+ - \phi^\star}^3 \\
        &\leq  \norm{\hat \phi^+ - \phi^\star}^2_{\calH(\phi^\star) + \nu I} +  C_{\cost} \norm{\hat \phi^+ - \phi^\star}^3 \\
        &\leq \paren{1+  \frac{C_{\cost} \norm{\hat\phi^+-\phi^\star}}{\lambda_{\min}(\calH(\phi^\star)) + \nu} } \norm{\hat \phi^+ - \phi^\star}^2_{\calH(\phi^\star) + \nu I}  \\
        &\leq  \paren{1+  \frac{ C_{\cost} R}{((1-\gamma) N)^{\alpha} \paren{\lambda_{\min}(\calH(\phi^\star)) + \nu} }} \frac{1+\xi}{1-\gamma} \paren{2 + 4 \log\frac{8}{\delta}} \frac{1}{N} \trace(\paren{\calH(\phi^\star)  + \nu I}\FI^{\pi_{\mix}}(\hat\phi^-)^{-1}) \\ & \mbox{(Events $\calE_{\mathsf{ID}}$ and the slow rate bound.)} \\
        & \leq \paren{1+  \frac{ C_{\cost} R}{((1-\gamma) N)^{\alpha} \paren{\lambda_{\min}(\calH(\phi^\star)) + \nu} }}  \frac{\frac{(1+\varepsilon)(1+\xi)}{(1-\gamma)^2} \paren{2 + 4 \log\frac{8}{\delta}}}{N} \Bigg(\inf_{\pi \in \Pi_{\exp}} \trace(\paren{\calH(\phi^\star)  + \nu I}\FI^{\pi}(\hat\phi^-)^{-1}) \\
        &+ \mathsf{poly}\paren{L_{\pi^\star}, L_f, L_{\theta}, L_{\cost}, \sigma_w^{-1}, \sigma_w, T, \dx, d_{\phi}, \norm{\calH(\phi^\star)}, \nu,  \frac{1}{\mu}, \frac{1}{\lambda_{\min}(\calH(\phi^\star)) + \nu}}  \frac{R}{(\gamma N)^\alpha}\Bigg).
    \end{align*}

    To conclude, we leverage the burn-in condition
    \begin{align*}
        N \geq \max \Bigg\{\paren{\frac{C_{\mathsf{H \,pert}}  R }{\beta \inf_{\pi \in \Pi_{\exp}} \trace\paren{\paren{\calH(\phi^\star) + \nu I } \FI^{\pi}(\phi^\star)^{-1}}}}^{1/\alpha}, \paren{ \frac{C_{\cost} R}{\beta (1-\gamma) \paren{\lambda_{\min}(\calH(\phi^\star)) + \nu}} }^{1/\alpha}\Bigg\},
    \end{align*}
    where 
    \begin{align*}
        C_{\mathsf{H \,pert}}  = \mathsf{poly}\paren{L_{\pi^\star}, L_f, L_{\theta}, L_{\cost}, \sigma_w^{-1}, \sigma_w, T, \dx, d_{\phi}, \norm{\calH(\phi^\star)}, \nu, \frac{1}{\mu}, \frac{1}{\lambda_{\min}(\calH(\phi^\star)) + \nu}}  
    \end{align*}
    Under this condition, the excess cost bound simplifies to
    \begin{align*}
        \calJ(\phi^+) - \calJ(\phi^\star) &\leq \frac{\frac{(1+\beta)^2 (1+\varepsilon) (1+\xi)}{(1-\gamma)^2} (2 + 4 \log\frac{8}{\delta})}{N}   \min_{\pi \in \Pi_{\exp}} \trace\paren{\paren{\calH(\phi^\star) + \nu}  \FI^{\pi}(\phi^\star)^{-1}}.
    \end{align*}
\end{proof}

\section{Aditional Experiments and Experiment Details}
\label{s: experiment details}

In \Cref{s: cartpole experiments}, we complement our illustrative example with an implementation of $\texttt{ALCOI}$ on a toy physical system. In \Cref{s: implementation details} we provide further experimental details. 

\subsection{Cartpole Experiments}
\label{s: cartpole experiments}

    We additionally run our proposed algorithm on the cartpole system defined by the dynamics 
    \begin{align*}
        (M+m) (\ddot p + b_p \dot p) + m \ell \cos(\theta) (\ddot \theta  + b_\theta \dot \theta) &= m \ell \dot \theta^2 \sin(\theta) + u \\
        m \cos(\theta) (\ddot p + b_p \dot p) + m\ell (\ddot \theta + b_\theta \dot \theta) &= m g \sin(\theta),
    \end{align*}
    where $p$ is the position of the cart, $\theta$ is the angle of the pole from the upright position, and $u$ is the control input. Here, $M$ represents the mass of the cart, $m$ the mass of the pole, $\ell$ the length of the pole, $g$ the acceleration due to gravity, $b_x$ is the friction coefficient for the cart, and $b_\theta$ is the friction coefficient for the pole. We discretize the system using the Euler approach using a timestep of $dt = 0.1$. We simluate the discretized system under additive zero mean Gaussian noise with covariance $0.01 I_4$. The unknown parameters are $\phi_\star = \bmat{M, m, \ell, g, b_x, b_\theta} = [1, 0.1, 1, 10, 0.5, 0.5]$ For every episode, the system starts from the hanging position with state $\bmat{0 & 0 & \pi & 0}$. The desired behavior is to swing the pole to the upright position with the cart positioned at the origin within the time horizon of $T=30$ timesteps. This is described by the cost functions $c_t(x,u) = 0$, $c_{T+1}(x) = \norm{x}^2$.

    To simplify the policy synthesis, rather than solving for the certainty equivalent policy which directly minimize the cost, we deploy energy shaping controllers which switch to an LQR controller about the upright position \citep{tedrake2009underactuated}. In particular, the certainty equivalent controller under estimates $\hat \phi = \bmat{\hat M, \hat m, \hat \ell, \hat g, \hat b_x, \hat b_\theta}$ selects the input as follows. 
    If $\abs{\theta} < \pi/4$, set $u$ according to the LQR controller synthesized by linearizing the dynamics under the estimated parameters about the upright equilibrium point with weights $Q=I$, $R=1$. Otherwise, set
    \begin{align}
        \label{eq: cartpole feedback linearization controller}
        u &= (M+m - m \cos(\theta)) \ddot x_d + m g \cos(\theta)\sin(\theta) - m \ell \dot \theta^2 \sin(\theta) \\
        \ddot x_d &= 5 \dot \theta \cos(\theta) \paren{\frac{1}{2} m \ell^2 \dot\theta^22 + m g \ell \cos(\theta) - m g \ell} - 0.01 p - 0.01 \dot p.
    \end{align}
    During exploration, the input energy budget is restricted to $0.1 T$. 

    The comparison between the three approaches, random exploration, approximate $A$-optimal exploration, and $\texttt{ALCOI}$ is shown in \Cref{fig: cartpole experiments}. As in the illustrative example, we see the benefit of control-oriented exploration over both random exploration. 

    \begin{figure}
        \centering
        \includegraphics[width=0.5\textwidth]{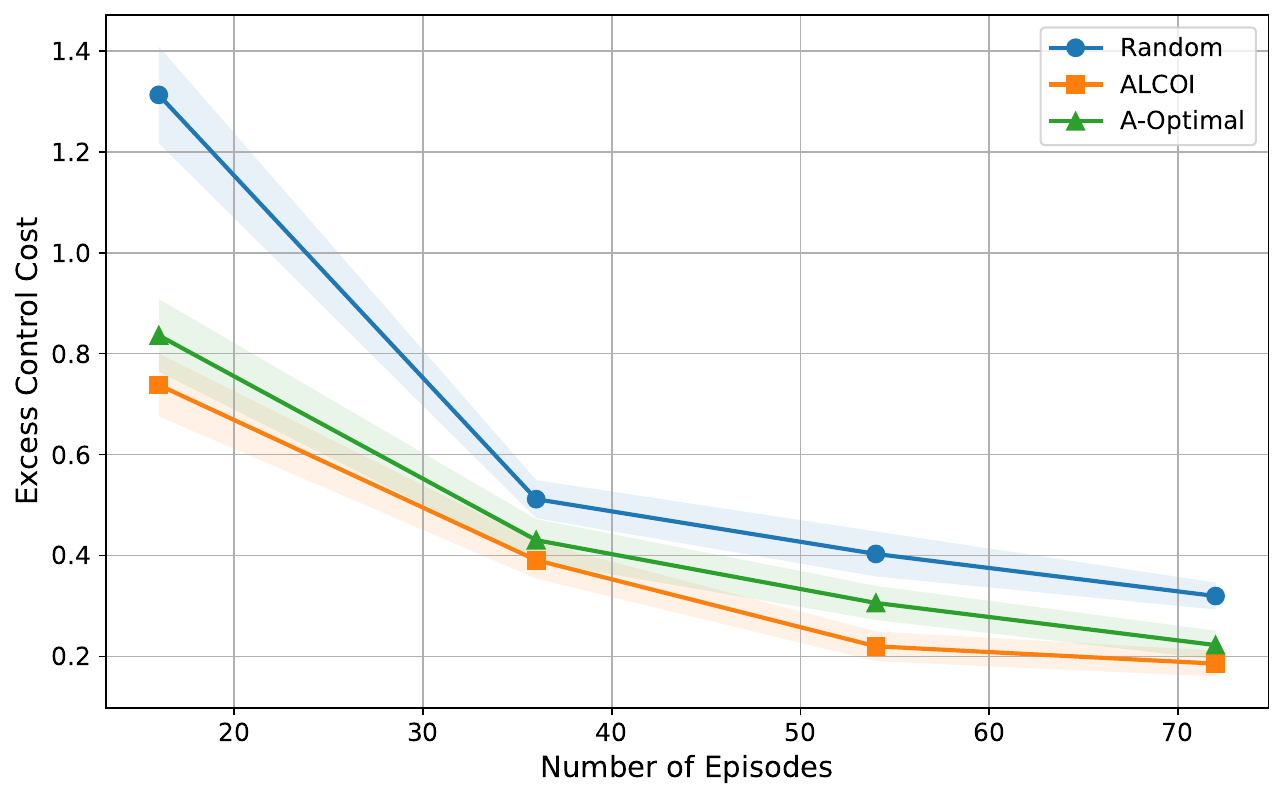}
        \caption{Excess cost versus number of exploration episodes for the proposed control-oriented identification procedure, approximate $A$-optimal design, and random exploration for the cartpole swing-up task. The mean over $900$ runs is shown, and the standard error is shaded.}
        \label{fig: cartpole experiments}
    \end{figure}

\subsection{Implementation Details}
\label{s: implementation details}

For both the cartpole experiments and the illustrative example, the initial exploration policy used by $\texttt{ALCOI}$ and approximate $A$-optimal exploration is set to the random exploration policy. 

We deviate from the proposed algorithm in the following respects. 
\begin{itemize}
    \item We use all of the collected data to fit the refined estimate.
    \item As the policy class in both experiments consists of all trajectory dependent policies with bounded input energy, we approximate this certainty equivalent synthesis via a receding horizon control procedure. In particular, at time step $t$ if the current state is $x_t$, the learner solves the problem:
\begin{align*}
   \underset{\tilde x_t,\dots, \tilde x_{T+1}, \tilde u_t, \dots, \tilde u_T}{\mathsf{minimize: }} & \quad \sum_{k=t}^{T} c(\tilde x_{k}, \tilde u_{k}) \\
   \mathsf{s.t.: } &\quad \tilde x_{k+1} = f(\tilde x_{k}, \tilde u_{k};\hat\phi), \quad \tilde x_t = x_t \quad \sum_{k=1}^t \norm{u_{k}}^2 +  \sum_{k=t+1}^T \norm{\tilde u_{k}}^2\leq \mathsf{total\,budget},
\end{align*}
where $\hat\phi$ is the learner's current esimate of the system parameters and $(\tilde x_k$, $\tilde u_k)$ constitute the planned trajectory. The learner plays the input $\tilde u_t$, on the true system, and re-evaluates the optimal control after observing the next state. The solution to the optimization problem defining the above receding horizon controller is computed via random shooting with a warm start. In particular, before applying each control input, $100$ new control sequences are generated from a standard normal distribution, and are normalized to satisfy the energy constraints. The seqeuence which minimizes the objective is then selected. This procedure is warm-started with the solution at the previous timestep by setting that as an additional candidate solution. 
\item The certainty equivalent policies intended to minimize the objective $\calJ(\pi, \phi^\star)$ are replaced with feedback linearization policies defined in terms of the parameter estimate. In particular, for the illustrative example and cartpole example, the policies are synthesized as 
\begin{enumerate}
    \item Illustrative example: $\hat \pi(\hat\phi) =  K \paren{X_t - \bmat{-5.5 \\ 0}} - \sum_{i=1}^4 \psi(X_t - \hat\phi^{(i)})$, where $K$ is the LQR controller synthesized via a linearization of the dynamics under the parameter $\hat \phi$. 
    \item Cartpole: the swing up controller is defined using feedback linearization as in \eqref{eq: cartpole feedback linearization controller}. 
\end{enumerate}
Note that a consequence of this is that there may exist policies in the policy class $\Pi_\star$ which achieve lower cost than the policy defined in terms of the optimal parameter $\phi_\star$. 
\end{itemize}
The first modification is made because discarding the a portion of the data is wasteful. It is included in the algorithm only in order to simplify the analysis, as the constant factor of $(1-\beta)$ is not a concern of this analysis. The sampling based approach is used since the optimization problem defining the receding horizon controller is intractable. The final modification is made to avoid the computational burden of solving the policy optimization problem via a Bellman equation in both policy synthesis and in computation of the model-task Hessian. These modifications result in a small gap between the theoretical guarantees, and the numerical validation, which future work will remove. 

\end{document}